\documentclass[preprint,10pt]{elsarticle}
\usepackage{url}
\usepackage{amsmath}
\usepackage{amssymb}
\usepackage{amsthm}
\usepackage{graphics}
\usepackage{graphicx}
\usepackage{grffile}
\usepackage{times}
\usepackage{algorithm,algorithmic}

\newcommand{\E}[1]{\mathbf{E}\left[#1\right]}

\newtheorem{theorem}{Theorem}
\newtheorem{corollary}[theorem]{Corollary}
\newtheorem{assumption}[theorem]{Assumption}
\newtheorem{definition}[theorem]{Definition}
\newtheorem{lemma}[theorem]{Lemma}

\urldef{\mailsa}\path|rafal.kapelko@pwr.edu.pl|

\begin{document}
\begin{frontmatter}
\title{
On the data persistency of replicated erasure codes in distributed storage systems 
}
\author[tech]{Roy Friedman\fnref{techfootnote}}
\ead{roy@cs.technion.ac.il}
\author[pwr]{Rafa\l{} Kapelko\corref{cor1}\fnref{pwrfootnote}}
\ead{rafal.kapelko@pwr.edu.pl}
\author[km]{Karol Marchwicki\fnref{kmfootnote}}
\ead{karol.marchwicki@gmail.com }
\fntext[techfootnote]{Research supported by the Israel Science Foundation grant \#1505/16.}
\fntext[pwrfootnote]{Supported by Polish National Science Center (NCN) grant 2019/33/B/ST6/02988}
\cortext[cor1]{Corresponding author at: Department of Fundamentals of Computer Science,
Wroc{\l}aw University of Technology, 
 Wybrze\.{z}e Wyspia\'{n}skiego 27, 50-370 Wroc\l{}aw, Poland. Tel.: +48 71 320 33 62; fax: +48 71 320 07 51.}
\address[tech]{Computer Science Department, Technion - Israel Institute of Technology, Haifa 32000, Israel}
\address[pwr]{ Department of Fundamentals of Computer Science, 
Wroc{\l}aw University of Technology, Poland}

\begin{abstract}
This paper studies the fundamental problem of \textit{data persistency} for a general 
family of redundancy schemes in distributed storage systems, called
\textit{replicated erasure codes.}
Namely, we analyze two strategies of replicated erasure codes distribution: \textit{random} and \textit{symmetric}.
For both strategies we derive closed analytical and asymptotic formulas for expected data persistency despite nodes \textit{failure}.
\end{abstract} 

\end{frontmatter}

\section{Introduction}
\label{sec:Introduction}

\subsection{Background}
\label{sec:background}
Distributed storage systems are very common these days~\cite{aghayev2020, bracher2019,  elyasi2020, rawat2016}.
Storing data on multiple servers enables scalability, i.e., the ability to handle a large number of clients, a large number of concurrent requests, and large amounts of data.
Further, geo-distribution helps reduce client access times by routing client's requests to a server that is relatively close to them in terms of Internet topology.
Finally, distributed systems offer an inherent level of fault-tolerance, since if one server fails or becomes disconnected, other servers can potentially step in to continue the service.

An important criteria for storage systems is long term data persistency.
That is, it should be possible to retrieve any stored data item even long after it has last been actively saved.
This is despite failures and even complete erasures of storage units, which may occur due to wear and tear, operator errors, bugs, and even malicious activity.
To that end, a distributed system offers two immediate mechanisms: \emph{replication} and \emph{error correcting codes}~\cite{LCL04,RB05,WK02}.
In replication, multiple copies of each data item are stored in multiple distinct servers which are assumed to exhibit independent failure pattern.
This way, if one server is down, or its storage has been erased, the data can still be retrieved from a replica.
This ensures fast recovery, at the cost of multiplying the storage overhead by the replication factor.

In error correction, each collection of $p$ data items is encoded into $p+q$ chunks, which are spread over $p+q$ distinct, failure independent, servers.
This way, when needed, a stored data item can be retrieved from any $p$ out of the corresponding $p+q$ servers that collectively encoded the missing data item.
This reduces the storage cost, but increases the recovery time and network overhead during data recovery.

Hence, replication and error correcting codes present a trade-off, whose right balance point depends on the applications requirements, network bandwidth, storage costs, and more~\cite{WK02}.
The \textit{expected data persistency} as well as the maintenance costs of any of the above systems may be analyzed based on the choice of the replication degree $r$ or the choice of $p$ and $q$, as the case may be, together with the other systems characteristics~\cite{LazySmart}.

In order to avoid the above dichotomy and offer a more flexible design, \textbf{replicated erasure codes} have been proposed~\cite{tech_kantor}.
In the latter, each of the $p+q$ chunks may be replicated $r$ times.
In this paper, we investigate the fundamental problem of \textit{expected data persistency} of \textbf{replicated erasure codes} when the storage nodes permanently leave the system and erase their locally stored data.
This study is related and motivated by the work of \cite{tech_kantor}. 
Namely, \textbf{the general family} of redundancy schemes for distributed storage systems, called replicated erasure codes
was introduced in~\cite{tech_kantor}. 
\subsection{Preliminaries and Notation}
\label{sec:preliminaries}
Formally, we recall the definition of \textit{replicated erasure codes}.
\begin{definition}[replicated erasure codes $REC(p,p+q,r)$]
\label{definition:first} Let $p\in\mathbb{N}\setminus\{0\},$ $q\in\mathbb{N}$ and $r\in\mathbb{N}\setminus\{0\}$ be fixed.
\begin{itemize}
\item The original file $f$ is divided into $p$ chunks and then encoded into $p+q$ chunks;
$f_1(1), f_2(1),$ $\dots f_{p+q}(1)$ chunks after encoding.
\item Each of the $p+q$ chunks is replicated into $r\ge 1$ replicas such that:\\
$f_1(1), f_1(2),\dots f_{1}(r)$ are replicas of the chunk $f_1(1)$,\\
$f_2(1), f_2(2),\dots f_{2}(r)$ are replicas of the chunk $f_2(1)$,\\ 
$\dots$\\
$f_{p+q}(1), f_{p+q}(2),\dots f_{p+q}(r)$ are replicas of the chunk $f_{p+q}(1).$
\item The total $r\cdot (p+q)$ chunks are spread according to some \textbf{replication strategy}
among storage nodes.
\end{itemize}
\end{definition}
The purpose of this paper is the analysis of two replication strategies: \textbf{random}~\cite{peertopeer} and \textbf{symmetric}~\cite{symmetric}. 
The \textbf{random strategy} is simple: all chunks are spread randomly among storage nodes.
In the \textbf{symmetric strategy} nodes in the system are split into groups of size $(p+q)\cdot r$
in which every node keeps chunks of all data stored at other nodes in the same group
(The precise meaning of the mentioned strategies
will be explained further in Algorithm~\ref{alg_rand} in Section~\ref{section:random}
and in Algorithm~\ref{alg_symnew} in Section~\ref{sec:symmetric}).

Also, we recall the properties of restoring an original data item that is divided and encoded according to erasure codes in terms of the replication multisets.
\begin{assumption}[properties of erasure codes $REC(p,p+q,r)$]
\label{assumption:first}
We call the multisets\\ $U_1=\{f_1(1), f_1(2),\dots f_{1}(r)\},$  $U_2=\{f_2(1), f_2(2),\dots f_{2}(r)\},$ 
 $\dots$\\ $U_{p+q}=\{f_{p+q}(1), f_{p+q}(2),\dots f_{p+q}(r)\}$  the replication multisets.
To restore the original file, one has to download $p$ \textbf{chunks}, each from a \textbf{different replication multiset.}
\end{assumption}
Figure~\ref{fig.multiset} illustrates Assumption~\ref{assumption:first} for system parameters $p=1,$ $q=2$ and $r=4$.
Observe that, in this special case, to restore the original file $f$, one has to download two fragments $f_1(i),$ $f_2(j)$
or $f_1(i),$ $f_3(j)$ or $f_2(i),$ $f_3(j)$ for some $i,j\in\{1,2,3,4\}$.
\begin{figure*}[ht]
\begin{center}
 \begin{tabular}{|c|cccc|} 
 \hline
 $U_1$ & $f_1(1)$  & $f_1(2)$ & $f_1(3)$ & $f_1(4)$  \\ [0.5ex]
 \hline
 $U_2$ & $f_2(1)$  & $f_2(2)$ & $f_2(3)$ & $f_2(4)$  \\ [0.5ex]
 \hline
 $U_3$ & $f_3(1)$  & $f_3(2)$ & $f_3(3)$ & $f_4(4)$  \\ [0.5ex]
 \hline
 \end{tabular}
 \end{center}
  \caption{The replication multisets $U_1,$ $U_2,$ $U_3$ for the 
  replicated erasure codes $REC(1,3,4),$ i.e. the 
  system parameters $p=1,$ $q=2$ and $r=4.$} 
  \label{fig.multiset}
\end{figure*}

We are now ready to formulate the main problem of \textbf{data persistency} of \textit{replicated erasure codes}
that is investigated in this paper.
 
\begin{assumption}[data persistency of ($REC(p,p+q,r)$)] 
\label{assumption:persistency}
Let $N$ be the number of nodes in the storage system and let $D$ be the number of documents.
\begin{itemize}
 \item[(i)] Assume that each document is replicated
according to replicated erasure codes (see Definition \ref{definition:first} and Assumption \ref{assumption:first}).
\item[(ii)] Suppose that nodes successively, independently at random leave the storage system and erase the data that they stored without notice.
\end{itemize}
Let $X_{(p,p+q,r)}^{(N,D)}$ denote the number of nodes which must be removed so that \textbf{restoring some document is no longer possible}.
The data persistency of $REC(p,p+q,r)$ is the random variable $X_{(p,p+q,r)}^{(N,D)}.$
\end{assumption}

Our main goal is to study the expected value of the random variable $X_{(p,p+q,r)}^{(N,D)}$ as a function of the parameters $p,q,r,N,D$.
Specifically, we derive a closed analytical formula for $\E{X_{(p,p+q,r)}^{(N,D)}}$ under two replication strategies: \textbf{random} and \textbf{symmetric},
and explain tradeoffs arising among the parameters $p,q,r,N,D$.

\subsection{Contributions and Outline of this Paper}
The main contributions of this work include deriving closed form and asymptotic formulas for the persistency of distributed storage systems that employ the replicated erasure coding scheme using the random and symmetric placement strategies.
For the specific case that the chosen parameters provide maximal persistency, we give an exact formula.
Our formulas are expressed in terms of Beta functions. 

More specifically, consider a distributed storage system consisting of $N$ nodes and $D$ documents, where each document is replicated according to replicated erasure codes.
That is, fix $p\in\mathbb{N}\setminus\{0\}$, $q\in\mathbb{N}$ and $r\in\mathbb{N}\setminus\{0\}$.
Each document is divided into $p$ chunks and then encoded into $p+q$ chunks. 
Each of the $p+q$ chunks are replicated into $r$ replicas (see Definition \ref{definition:first}).
All chunks are spread in the storage nodes according to either the random or the symmetric strategy (see Algorithm~\ref{alg_rand} in Section~\ref{section:random} and Algorithm~\ref{alg_symnew} in Section~\ref{sec:symmetric}).

We further assume that nodes successively, independently at random leave the storage system and erase the data that they stored without notice.
The objective is to derive the formula for the expected data persistency of replicated erasure codes in the storage system, i.e.,  the number of nodes in the storage system which must be removed so that restoring some document is no longer possible.
To this aim, we make the following two novel contributions:
\begin{itemize}
\item We derive for \textbf{random} and \textbf{symmetric} replication strategies closed form integral formula for the expected data persistency in terms of incomplete Beta function
 (see Theorem \ref{thm:mainbex} in Section \ref{section:random} and Theorem \ref{thm:mainasymerfirst} in Section \ref{sec:symmetric})
 and asymptotics formulas (see Theorem \ref{thm:mainbd} in Section \ref{section:random} and Theorem \ref{thm:mainasymer} in Section \ref{sec:symmetric}).
 \item When the maximal expected data persistency of replicated erasure codes 
for \textbf{both replication strategies} is attained by the replicated erasure codes with parameter $p=1$
(see Corollary \ref{thm:mainbdef} in Section \ref{section:random} 
and Corollary \ref{thm:mainbdefabc} in Section \ref{sec:symmetric}) 
\footnote{Notice that for $p=1$ replicated erasure codes becomes
simply replication with $(1+q)r$ replicas.}
we give exact formulas for the expected data persistency in terms of Beta Function (see Subsection \ref{sub:exactly}).
\end{itemize}
Table \ref{tab:mainaab} summarizes the results proved in this paper (see the Landau asymptotic notation in Section \ref{sec:preliminar}).
Let us recall that to restore the original file, one has to download $p$ chunks, each from $r$ different replication sets (see Assumption \ref{assumption:persistency}).
Thus, if $q+1$ chunks in each $r$ replication sets are distributed in the nodes that leave the storage system then restoring the original file is not possible. 
For the random replication strategy, the expected data persistency is in $N\Theta\left(D^{-\frac{1}{r(q+1)}}\right)$.
If the documents are replicated according to symmetric strategy the expected data persistency is in $\Theta\left(N^{1-\frac{1}{r(q+1)}}\right)$.

\begin {table*}[t]
\caption {The expected data persistency of replicated erasure codes as a function of the replication parameters $p,q,r,$ the number of nodes $N$ and the number of documents $D.$}
\label{tab:minaaa} 
\begin{center}
 \label{tab:mainaab}
 \begin{tabular}{|*{4}{c|}} 
 \hline
 Replication strategy  & $\E{X_{(p,p+q,r)}^{(N,D)}}$ & Algorithm  & Theorems \\ [0.5ex]
    \hline
 random & $N\Theta\left(D^{-\frac{1}{r(q+1)}}\right)$  & \ref{alg_rand} &    \ref{thm:mainbex}, \ref{thm:mainbd} \\ [0.5ex]
 \hline
symmetric  & $\Theta\left(N^{1-\frac{1}{r(q+1)}}\right)$ & \ref{alg_symnew} &  \ref{thm:mainasymerfirst}, \ref{thm:mainasymer} \\  [0.5ex]
 \hline
 \end{tabular}
\end{center}
\end{table*}

The rest of the paper is organized as follows. Section \ref{sec:preliminar}  provides several preliminary results that will be used in the sequel.
In Section \ref{section:random} we derive the main results on the expected data persistency of replicated erasure codes for the random strategy.
Section \ref{sec:symmetric} deals with the main results on the expected data persistency of replicated erasure codes for the symmetric strategy.
Section \ref{sec:discussions}  presents further insights including \textbf{comparision of the random and symmetric placement strategies}, \textit{non-uniform redundancy scheme} and \textit{exact formulas} in terms of Beta function. 
The numerical results are presented in Section \ref{sec:experiments}.
Finally, Conclusions are drawn in Section~\ref{sec:conclusions}, where we also outline some future research directions.
\section{Mathematical Background}
\label{sec:preliminar}
In this section, we recall some basic facts about special functions: Gamma, Beta and incomplete Beta that will be used
in the sequel. We also prove Lemma~\ref{lemma:pochodna1} and Lemma~\ref{lemma:pochodna2} about asymptotic expansions for some function that will be helpful in deriving our main asymptotic results Theorem~\ref{thm:mainbd} in Section~\ref{section:random} and Theorem~\ref{thm:mainasymer} in
Section~\ref{sec:symmetric}. 
Finally, we recall Formula~(\ref{eq:random_variable}) for discrete nonnegative random variable and the Landau asymptotic
notation.
The Gamma function is defined via a convergent improper integral
$$
\Gamma(z)=\int_0^{\infty}t^{z-1}e^{-t}dt,\,\,\,\,z>0.
$$
If $n$ is an integer, the Gamma function is related to factorial
$$
\Gamma(n+1)=n!.
$$
The Beta function $\mathrm{Beta}(a,b)$ (see~\cite[Chapter 8]{stat_2011}) is defined for all positive integer numbers $a,b$ via a convergent improper integral
\begin{equation}
\label{eq:first01}
\mathrm{Beta}(a,b)=\int_0^1 t^{a-1}(1-t)^{b-1}dt.
\end{equation}
It is related to the binomial coefficient. Namely, 
\begin{equation}
\label{eq:first02}%
\mathrm{Beta}(a,b)=\frac{1}{\binom{a+b-2}{a-1}(a+b-1)}.
\end{equation}
The incomplete Beta function is defined for all positive integer numbers $a,b$ by the formula
\begin{equation}
\label{eq:incomplete01}
I_x(a,b)=\frac{1}{\mathrm{Beta}(a,b)}\int_0^x t^{a-1}(1-t)^{b-1}dt,\,\,\, \text{when}\,\,0\le x\le 1.
\end{equation}
The following Identity about incomplete Beta function is known and can be proved using integration by parts.  
\begin{equation}
\label{eq:incomplete02}
1-I_x(a,b)=\sum_{j=0}^{a-1}\binom{a+b-1}{j}x^j(1-x)^{a+b-1-j},
\end{equation}
(see \cite[Identity 8.17.5]{NIST} for $c=m$ and $d=n-m+1$).
\begin{lemma}
\label{lemma:pochodna1}
Assume that 
$p,r\in\mathbb{N}\setminus\{0\},$  $q\in\mathbb{N}$ are fixed and independent on $n.$
Consider the following function 
$$f(x)=- \ln\left(1-I_{x^r}(q+1,p)\right),$$ 
where $I_{x^r}(q+1,p)$ is the incomplete Beta function.
Then 
\begin{enumerate}
\item[(i)] $f(x)>f(0)$ for all $x\in(0,1),$ and for every $\delta>0$ the infimum of $f(x)-f(0)$ in $[\delta,1)$ is positive,
\item[(ii)] $f(x)$ has as $x\rightarrow 0^+$ the following expansion
$$\sum_{k=0}^{\infty}a_kx^{k+r(q+1)},$$
where $a_0=\binom{p+q}{q+1}$ and the expansion of $f$ is term-wise differentiated, that is
$f^{'}(x)=\sum_{k=0}^{\infty}a_k(k+r(q+1))x^{k+r(q+1)-1}.$
\end{enumerate}
\end{lemma}
\begin{proof} 
Using Formula~(\ref{eq:incomplete01}) we see that the function $I_{x^r}(q+1,p)$ is monotonically increasing over the interval $[0,1]$, giving us the first part of Lemma~\ref{lemma:pochodna1}.

To prove the second part, we apply~(\ref{eq:incomplete02}) for $a:=q+1,$ $b:=p,$ $x:=x^r$ and derive
$$f(x)=- \ln\left(\sum_{j=0}^{q}\binom{p+q}{j}x^{rj}\left(1-x^r\right)^{p+q-j}\right).$$ 
By the binomial theorem we have
\begin{align*}
f(x)&=- \ln\left(1-\sum_{j=q+1}^{p+q}\binom{p+q}{j}x^{rj}\left(1-x^r\right)^{p+q-j}\right)\\
&=- \ln\left(1-\binom{p+q}{q+1}x^{r(q+1)}-\sum_{j=q+2}^{p+q}a_jx^{rj}\right).
\end{align*}
It is very well known that the natural logarithm has Maclaurin series 
$$
-\ln(1-x)=\sum_{k=1}^{\infty}\frac{x^k}{k}.
$$
Applying this expansion for $x:=\binom{p+q}{q+1}x^{r(q+1)}+\sum_{j=q+2}^{p+q}a_jx^{rj}$ we get
\begin{align*}
&-\ln\left(1-\binom{p+q}{q+1}x^{r(q+1)}-\sum_{j=q+2}^{p+q}a_jx^{rj}\right)\\
&=\sum_{k=1}^{\infty}\frac{\left(\binom{p+q}{q+1}x^{r(q+1)}+\sum_{j=q+2}^{p+q}a_jx^{rj}\right)^k}{k}.
\end{align*}
From this we derive easily
\begin{align*}
f(0)=&f'(0)=f''(0)=\dots=f^{(r(q+1)-1)}(0)=0,\\
&f^{(r(q+1))}(0)= \binom{p+q}{q+1}(r(q+1))!
\end{align*}
Therefore
$$f(x)=\sum_{k=0}^{\infty}a_kx^{k+r(q+1)},\,\,\,\,\,\, f^{'}(x)=\sum_{k=0}^{\infty}a_k(k+r(q+1))x^{k+r(q+1)-1}$$ 
and 
$$
a_0=\frac{f^{(r(q+1))}(0)}{(r(q+1))!}=\binom{p+q}{q+1}.
$$
This is enough to prove Lemma \ref{lemma:pochodna1}.
\end{proof}
\begin{lemma}
\label{lemma:pochodna2}
Assume that 
$p,r\in\mathbb{N}\setminus\{0\},$  $q\in\mathbb{N}$
are fixed and independent on $n.$
Consider the following function 
$$f(x)=-((p+q)r)^{-1} \ln\left(1-\left(I_{x}(q+1,p)\right)^r\right),$$ 
where $I_{x}(q+1,p)$ is the incomplete Beta function.
Then 
\begin{enumerate}
\item[(i)] $f(x)>f(0)$ for all $x\in(0,1),$ and for every $\delta>0$ the infimum of $f(x)-f(0)$ in $[\delta,1)$ is positive,
\item[(ii)] $f(x)$ has as $x\rightarrow 0^+$ the following expansion
$$\sum_{k=0}^{\infty}a_kx^{k+r(q+1)},$$
where $a_0=((p+q)r)^{-1}\binom{p+q}{q+1}^r$ and the expansion of $f$ is term-wise differentiated, that is
$f^{'}(x)=\sum_{k=0}^{\infty}a_k(k+r(q+1))x^{k+r(q+1)-1}.$
\end{enumerate}
\end{lemma}
\begin{proof} 
The proof of Lemma~\ref{lemma:pochodna2} is analogous to that of Lemma~\ref{lemma:pochodna1}.
Using Formula~(\ref{eq:incomplete01}), we see that the function $I_{x}(q+1,p)$ is monotonic increasing over the interval $[0,1]$ and get the first part of Lemma \ref{lemma:pochodna1}.

To prove the second part, we apply (\ref{eq:incomplete02}) for $a:=q+1,$ $b:=p,$ $x:=x$ and derive
$$
f(x)=-((p+q)r)^{-1} \ln\left(1-\left(1-\sum_{j=0}^{q}\binom{p+q}{j}x^{j}\left(1-x\right)^{p+q-j}\right)^r\right).
$$
By the binomial theorem we have
\begin{align*}
f(x)&=-((p+q)r)^{-1} \ln\left(1-\left(\sum_{j=q+1}^{p+q}\binom{p+q}{j}x^{j}\left(1-x\right)^{p+q-j}\right)^r\right)\\
&=-((p+q)r)^{-1} \ln\left(1-\binom{p+q}{q+1}^rx^{r(q+1)}-\sum_{j=q+2}^{p+q}a_jx^{rj}\right).
\end{align*}
Let us recall that the natural logarithm has Maclaurin series 
$$
-\ln(1-x)=\sum_{k=1}^{\infty}\frac{x^k}{k}.
$$
Using this for $x:=\binom{p+q}{q+1}^rx^{r(q+1)}+\sum_{j=q+2}^{p+q}a_jx^{rj}$ we get
\begin{align*}
&-\ln\left(1-\binom{p+q}{q+1}^rx^{r(q+1)}-\sum_{j=q+2}^{p+q}a_jx^{rj}\right)\\
&=\sum_{k=1}^{\infty}\frac{\left(\binom{p+q}{q+1}^rx^{r(q+1)}+\sum_{j=q+2}^{p+q}a_jx^{rj}\right)^k}{k}.
\end{align*}
From this we derive easily
\begin{align*}
f(0)=&f'(0)=f''(0)=\dots=f^{(r(q+1)-1)}(0)=0,\\
&f^{(r(q+1))}(0)=((p+q)r)^{-1}\binom{p+q}{q+1}^r(r(q+1))!
\end{align*}
Therefore
$$f(x)=\sum_{k=0}^{\infty}a_kx^{k+r(q+1)},\,\,\,\,\,\, f^{'}(x)=\sum_{k=0}^{\infty}a_k(k+r(q+1))x^{k+r(q+1)-1}$$ 
and 
$$
a_0=\frac{f^{(r(q+1))}(0)}{(r(q+1))!}=((p+q)r)^{-1}\binom{p+q}{q+1}^r.
$$
This is enough to prove Lemma \ref{lemma:pochodna2}.
\end{proof}
Assume that $X$ is a discrete nonnegative random variable. Then
\begin{equation}
\label{eq:random_variable}
\E{X}=\sum_{l\ge 0}\Pr[X>l]
\end{equation}
(see \cite[Exercise 2.1]{Szpankowski2001}).

We recall the following asymptotic notation:\\
$(a)$ $f(n)=O(g(n))$ if there exists a constant $C_1>0$ and integer $N$ such that $|f(n)|\le C_1|g(n)|$ for all $n>N,$\\
$(b)$ $f(n)=\Omega(g(n))$ if there exists a constant $C_2>0$ and integer $N$ such that $|f(n)|\ge C_2|g(n)|$ for all $n>N,$\\
$(c)$ $f(n)=\Theta(g(n))$ if and only if $f(n)=O(g(n))$ and $f(n)=\Omega(g(n)).$
\section{Analysis of Random Replicated Erasure Codes }
\label{section:random}
In this section, we investigate the expected data persistency of replicated erasure codes for the \textbf{random replication strategy.}
Namely, we analyse Algorithm~\ref{alg_rand}. 
It is worth pointing out that while Algorithm~\ref{alg_rand} is very simple, its analysis is not completely trivial.
\begin{algorithm}[tb]
\caption{Random replication of $REC(p,p+q,r)$}
\label{alg_rand}
\begin{algorithmic}[1]
 \REQUIRE A system of $N$ storage nodes; $f^{(1)}, f^{(2)},\dots f^{(D)}$ data items; each of the data items is encoded and replicated into $(p+q)r$ chunks (see Definition \ref{definition:first});
 $f^{(1)}_1(j),f^{(1)}_2(j),\dots f^{(1)}_{p+q}(j),$ $f^{(2)}_1(j),f^{(2)}_2(j),\dots f^{(2)}_{p+q}(j),\dots$ $f^{(d)}_1(j),f^{(d)}_2(j),\dots f^{(D)}_{p+q}(j),$
 for $j\in\{1,2\dots, r\},$
 fragments after encoding and replication.
 \ENSURE  The final random and independent positions of $(p+q)rD$ fragments in the nodes of the storage system (see Figure $2$).
 \FOR{$k=1$  \TO $D$ } 
  \FOR{$j=1$  \TO $r$ } 
   \FOR{$l=1$  \TO $p+q$ } 
  \STATE{choose randomly and independently node in the storage system;}
 \STATE{replicate fragment $f^{(k)}_{l}(j)$ at the chosen node;}
\ENDFOR
 \ENDFOR
 \ENDFOR
\end{algorithmic}
\end{algorithm}
Figure $2$ illustrates Algorithm \ref{alg_rand} for the system parameters $p=q=1$, $r=2$, $D=6$ and $N=8$.
In this case, each document is divided and encoded into $4$ chunks. 
Hence, we have $24$ chunks spread randomly and independently in the $8$ nodes.
\begin{figure*}[b]
  $$
\begin{array}{|c|c|c|c|c|c|c|c|c|c|}
\hline
 n & 1 & 2 & 3  & 4 & 5 & 6 & 7 & 8  \\
 \hline
 &f^{(1)}_1(1) & f^{(2)}_1(1) & f^{(1)}_2(1)  & f^{(1)}_1(2) & f^{(1)}_2(2) & f^{(2)}_2(1) 
  & f^{(2)}_1(2) 
 & f^{(2)}_2(2) \\
 
 &f^{(6)}_1(1)  & f^{(3)}_1(1) & f^{(4)}_1(2)  & f^{(3)}_2(2) & f^{(4)}_2(1) & f^{(4)}_1(1) 
  & f^{(3)}_2(1) 
 & f^{(3)}_1(2) \\
 
  & &  &   & f^{(4)}_2(2) & f^{(5)}_1(1) &  f^{(6)}_1(2)
  & f^{(5)}_1(2)
 & \\
 
  & &  &   &  & f^{(5)}_2(1) &  f^{(6)}_2(1)
  &  
 & \\
 
  & &  &   &  & f^{(5)}_2(2) &  
  & 
 &  \\
 
   & &  &   &  & f^{(6)}_2(2) &  
  & 
 &  \\
 \hline
\end{array}
$$ 
  \caption{Random replicated strategy according to Algorithm \ref{alg_rand} for
  replicated erasure codes $REC(1,2,2),$ i.e. the 
  system parameters $p=q=1,$ $r=2,$ the number of documents $d=6$ and the number of nodes $n=8.$} 
  \label{fig.1as}
\end{figure*}

We now prove the following exact formula for the expected data persistency of replicated erasure codes with the random replication strategy. 
In the proof of Corollary~\ref{thm:random_a}, 
we derive $\Pr\left[X_{(p,p+q,r)}^{(N,D)}>l\right]$ (see Equation~(\ref{eq:probas}))
and apply Identity~(\ref{eq:random_variable}) to get the desired exact formula as the sum of incomplete Beta functions.

\begin{corollary}
\label{thm:random_a}
Let Assumption \ref{assumption:persistency} hold. 
Assume that all chunks are spread randomly and independently in the storage nodes as in Algorithm~\ref{alg_rand}.
Then
$$
\E{X_{(p,p+q,r)}^{(N,D)}}=\sum_{l=0}^{N}\left(1-I_{\left(\frac{l}{N}\right)^r}(q+1,p)\right)^{D},
$$
where $I_{\left(\frac{l}{N}\right)^r}(q+1,p)$ is the incomplete Beta function.
\end{corollary}
\begin{proof} 
 Let $N$ be the number of nodes  and let $D$ be the number of documents in the storage system.
Assume that nodes successively, independently at random leave the storage system and erase the data that they stored without notice.

Let us recall that random variable $X_{(p,p+q,r)}^{(N,D)}$ denotes the number of nodes which must be removed so that restoring some document becomes impossible.
Fix $i\in\{1,2,\dots,p+q\}$.
Assume that the set of nodes of cardinality $l$ are all removed.
Since the $l$ nodes successively, independently at random leave the storage system, the probability that we remove all chunks from the replication set $U_i$ is equal to $\left(\frac{l}{N}\right)^r$. 
Notice that to restore the fixed file $f$, one has to download $p$ chunks, each from a different replication multiset, i.e., at most $q$ replication multisets can be removed
(see Assumption \ref{assumption:first}).
Therefore, the probability for restoring the fixed file $f$, when $l$ nodes successively, independently at random leave the storage system the storage system is equal to
$$
\sum_{j=0}^{q}\binom{p+q}{j}\left(\left(\frac{l}{N}\right)^r\right)^j\left(1-\left(\frac{l}{N}\right)^r\right)^{p+q-j} .
$$
Since each of $D$ documents are transformed into chunks (see Definition~\ref{definition:first}) and all the
chunks are spread independently and randomly among storage nodes, we have
\begin{equation}
\label{eq:probas}
\Pr\left[X_{(p,p+q,r)}^{(N,D)}>l\right]=\left(\sum_{j=0}^{q}\binom{p+q}{j}\left(\left(\frac{l}{N}\right)^r\right)^j\left(1-\left(\frac{l}{N}\right)^r\right)^{p+q-j}\right)^{D} .
\end{equation}
Combining together Equation~(\ref{eq:random_variable}) for random variable $X:=X_{(p,p+q,r)}^{(N,D)}$, Equation~(\ref{eq:probas}), as well as $\Pr[X_{(p,p+q,r)}^{(N,D)}> l]=0$ for any $l\ge N+1$,
we get
$$\E{X_{(p,p+q,r)}^{(N,D)}}=\sum_{l=0}^{N}\left(\sum_{j=0}^{q}\binom{p+q}{j}\left(\left(\frac{l}{N}\right)^r\right)^j\left(1-\left(\frac{l}{N}\right)^r\right)^{p+q-j}\right)^{D} .$$
Applying Identity~(\ref{eq:incomplete02}) for $a:=q+1,$ $b:=p$ we have the desired Formula
$$\E{X_{(p,p+q,r)}^{(N,D)}}=\sum_{l=0}^{N}\left(1-I_{\left(\frac{l}{N}\right)^r}(q+1,p)\right)^{D} .$$
This completes the proof of Corollary \ref{thm:random_a}.
\end{proof}
We now observe that the maximal expected data persistency of replicated erasure codes 
for the random replication strategy
is attained by the replicated erasure codes with parameter $p=1$.
Namely, we combine together the result of previous Corollary \ref{thm:random_a}
with well known Identity~\cite[Identity 8.17.21]{NIST} for the incomplete Beta function to obtain Corollary~\ref{thm:mainbdef}. 
\begin{corollary}
\label{thm:mainbdef}
Under the assumption of Corollary \ref{thm:random_a} we have
$$
 \max_{p\ge 1} \E{X_{(p,p+q,r)}^{(N,D)}} =\E{X_{(1,1+q,r)}^{(N,D)}}.
$$
\end{corollary}
\begin{proof} 
Applying
\cite[Identity 8.17.21]{NIST} for $a:=q+1$ and $b:=p$, as well as Identity (\ref{eq:first02})
for $a:=q+1,$ and $b:=p$ we have
\begin{equation}
\label{align:b100}
I_x(q+1,p+1)-I_x(q+1,p)=x^{1+q}(1-x)^p{p+q\choose p}\ge 0.
\end{equation}
Using (\ref{align:b100}) for $x:=\left(\frac{l}{N}\right)^r$
we have
$$
1-I_{\left(\frac{l}{N}\right)^r}(q+1,p)\ge 1-I_{\left(\frac{l}{N}\right)^r}(q+1,p+1).
$$
Hence
$$
\sum_{l=0}^{N}\left(1-I_{\left(\frac{l}{N}\right)^r}(q+1,p)\right)^D\ge \sum_{l=0}^{N} \left(1-I_{\left(\frac{l}{N}\right)^r}(q+1,p+1)\right)^D.
$$
Applying Corollary \ref{thm:random_a} we get
$$\E{X_{(p,p+q,r)}^{(N,D)}}\ge \E{X_{(p+1,p+1+q,r)}^{(N,D)}}.$$
Hence 
$$
 \max_{p\ge 1} \E{X_{(p,p+q,r)}^{(N,D)}} =\E{X_{(1,1+q,r)}^{(N,D)}}.
$$
This completes the proof of Corollary \ref{thm:mainbdef}.
\end{proof}
The next theorem provides an explicit integral formula for
$\E{X_{(p,p+q,r)}^{(N,D)}}$ 
\begin{theorem} 
\label{thm:mainbex}
Let Assumption \ref{assumption:persistency} hold. Assume that all chunks are spread randomly
and independently in the storage nodes as in  Algorithm \ref{alg_rand}.
Then
$$
  \E{X_{(p,p+q,r)}^{(N,D)}} =N\int_{0}^{1}\left(1-I_{x^r}(q+1,p)\right)^{D} dx+ER, 
$$
where $I_{x^r}(q+1,p)$ is the incomplete Beta function and $|ER|\le 1.$
\end{theorem}
\begin{proof} 
The main idea of the proof is simple. 
It is an approximation of the sum in Corollary~\ref{thm:random_a}
by an integral. Notice that
\begin{equation}
\sum_{l=0}^{N}\left(1-I_{\left(\frac{l}{N}\right)^r}(q+1,p)\right)^D
\label{eq:assa02}=\int_0^N\left(1-I_{\left(\frac{y}{N}\right)^r}(q+1,p)\right)^D   dy+\Delta,
\end{equation}
where $f(y)=\left(1-I_{\left(\frac{y}{N}\right)^r}(q+1,p)\right)^D$ and
$|\Delta|\le\sum_{l=0}^{N}\max_{l\le y<l+1}|f(y)-f(l)|$ 
(see~\cite[Page 179]{sedgewick}).
From~(\ref{eq:incomplete01}),
we deduce that the function $I_{\left(\frac{y}{N}\right)^r}(q+1,p)$ is monotone increasing over the interval $[0,N]$.
Therefore, the function $f(y)$ is monotone decreasing over the interval $[0,N]$ and the error term $|\Delta|$ telescopes on the interval $[0,N]$.
Hence
\begin{equation}
\label{eq:telescop}
|\Delta|\le |f(0)-f(N)|^d=1.
\end{equation}
It is easy to see that
\begin{equation}
\int_0^N \left(1-I_{\left(\frac{y}{N}\right)^r}(q+1,p)\right)^D  dy\\
\label{eq:fish}=N\int_0^1  \left(1-I_{x^r}(q+1,p)\right)^D  dx.
\end{equation}
Finally, combining together the result of Theorem~\ref{thm:mainbdef}
with~(\ref{eq:assa02}-\ref{eq:fish}) we get
$$
 \E{X_{(p,p+q,r)}^{(N,D)}} =N\int_{0}^{1}\left(1-I_{x^r}(q+1,p)\right)^D dx+ER, 
$$
where $|ER|\le 1$.
This is enough to prove Theorem~\ref{thm:mainbex}.
\end{proof}
Unfortunately, the formula in the statement of Theorem~\ref{thm:mainbex}
makes it difficult to apply it in practice.
To that end, we now formulate a very useful asymptotic approximation for $\E{X_{(p,p+q,r)}^{(N,D)}}$. 
In the next Theorem, we analyze the expected data persistency of replicated erasure codes for the random replication strategy when the number of documents $D$ is large. 
In the proof of Theorem~\ref{thm:mainbd}, we apply
the Erd\'elyi's formulation of Laplace's method for the integral formula obtained in Theorem~\ref{thm:mainbex}. The following asymptotic result is true 
\begin{theorem} 
\label{thm:mainbd}
Under the assumption of Theorem \ref{thm:mainbex} for large number of documents $D$ we have
$$
  \E{X_{(p,p+q,r)}^{(N,D)}} =\frac{\Gamma\left(1+\frac{1}{r(q+1)}\right)}{\binom{p+q}{q+1}^{\frac{1}{r(q+1)}}}N{D^{-\frac{1}{r(q+1)}}}\left(1+O\left(D^{-\frac{1}{r(q+1)}}\right)\right)+O(1),
$$
where $\Gamma\left(1+\frac{1}{r(q+1)}\right)$ is the Gamma function.
\end{theorem}
\begin{proof} 
Assume that $p\in\mathbb{N}\setminus\{0\},$  $q\in\mathbb{N}$ and $r\in\mathbb{N}\setminus\{0\}$ are fixed and independent of $D$.
Let 
\begin{equation}
\label{eq:lestac}
I(D)=\int_{0}^{1}\left(1-I_{x^r}(q+1,p)\right)^Ddx.
\end{equation}
We now provide asymptotic analysis of $I(D)$ for large $D.$ In this analysis we use the Erd\'elyi's formulation of Laplace's method. See \cite{Erdelyi, Nemes2013, Wojdylo} for a description of this technique.
Observe that
$$I(D)=\int_{0}^{1}e^{-Df(x)}dx,$$
where $f(x)=-\ln\left(1-I_{x^r}(q+1,p)\right).$

From Lemma~\ref{lemma:pochodna1} we deduce that the assumption of Theorem~1.1 in~\cite{Nemes2013} holds for $f(x):=-\ln\left(1-I_{x^r}(q+1,p)\right),$ $g(x):=1,$  $a:=0$ and $b:=1.$
Namely, we apply Formula~$(1.5)$ in~\cite[Theorem 1.1]{Nemes2013} 
for  $a:=0$ $f(0):=0,$ $\lambda:=D,$ $\beta:=1,$ $\alpha:=r(q+1),$  $b_0=1,$  $a_0=\binom{p+q}{q+1},$ $n:=k$  and deduce that the integral $I(D)$ has the following asymptotic expansion
$$\sum_{k=0}^{\infty}\Gamma\left(\frac{k+1}{r(q+1)}\right)\frac{c_k}{D^{\frac{k+1}{r(q+1)}}},$$
as $D\rightarrow\infty,$ where $(c_0)^{-1}=\binom{p+q}{q+1}^{\frac{1}{r(q+1)}}r(q+1)$ and $\Gamma\left(\frac{k+1}{r(q+1)}\right)$ is the Gamma function.
Therefore
$$I(D)=\frac{\Gamma\left(\frac{1}{r(q+1)}\right)\frac{1}{r(q+1)}}{{\binom{p+q}{q+1}}^{\frac{1}{r(q+1)}}}\frac{1}{D^{\frac{1}{r(q+1)}}}+O\left(\frac{1}{D^{\frac{2}{r(q+1)}}}\right).$$
Using the basic identity for the Gamma function 
$$\Gamma\left(\frac{1}{r(q+1)}\right)\frac{1}{r(q+1)}=\Gamma\left(1+\frac{1}{r(q+1)}\right)$$ 
(see \cite[Identity 5.5.1]{NIST} for $z:=\frac{1}{r(q+1)}$)
we get 
\begin{equation}
\label{eq:lesta}
I(D)=\frac{\Gamma\left(1+\frac{1}{r(q+1)}\right)}{{\binom{p+q}{q+1}}^{\frac{1}{r(q+1)}}}\frac{1}{D^{\frac{1}{r(q+1)}}}+O\left(\frac{1}{D^{\frac{2}{r(q+1)}}}\right).
\end{equation}
Finally, combining together the result of Theorem \ref{thm:mainbex}, (\ref{eq:lestac}) and (\ref{eq:lesta})
we have 
$$
  \E{X_{(p,p+q,r)}^{(N,D)}} =\frac{\Gamma\left(1+\frac{1}{r(q+1)}\right)}{\binom{p+q}{q+1}^{\frac{1}{r(q+1)}}}N{D^{-\frac{1}{r(q+1)}}}\left(1+O\left(D^{-\frac{1}{r(q+1)}}\right)\right)+O(1).
$$
This completes the proof of Theorem \ref{thm:mainbd}.
\end{proof}
Table~\ref{tab:main} explains Theorem~\ref{thm:mainbd} for $p\in\mathbb{N}\setminus\{0\}$ and some fixed parameters $q,r$ in terms of asymptotic notations. 
In view of Theorem~\ref{thm:mainbd}, we have $$\E{X_{(p,p+q,r)}^{(N,D)}}=N{\Theta\left(D^{-\frac{1}{r(q+1)}}\right)}+O(1).$$
\begin {table*}[!ht]
\caption {The expected data persistency of replicated erasure codes replicated erasure codes for \textbf{random replication strategy}
as a function of system parameters $p,q,r$ provided that $D$ is large.}
\label{tab:line} 
\begin{center}
 \label{tab:main}
 \begin{tabular}{|*{4}{c|}} 
 \hline
 parameter $p$  & parameter $q$ & parameter $r$  & $\E{X_{(p,p+q,r)}^{(N,D)}}$ \\ [0.5ex]
    \hline
 $p\in\mathbb{N}\setminus\{0\}$ & $q=0$  & $r=1$   & ${N}{\Theta\left(D^{-1}\right)}+O(1)$ \\ [0.5ex]
 \hline
 $p\in\mathbb{N}\setminus\{0\}$ & $q=0$ & $r=2$ &   $N{\Theta\left(D^{-\frac{1}{2}}\right)}+O(1)$ \\  [0.5ex]
 \hline
 $p\in\mathbb{N}\setminus\{0\}$ &$q=1$ & $r=2$ &  $N{\Theta\left(D^{-\frac{1}{4}}\right)}+O(1)$ \\  [0.5ex]
\hline
$p\in\mathbb{N}\setminus\{0\}$ & $q=2$ & $r=3$ &   $N{\Theta\left(D^{-\frac{1}{9}}\right)}+O(1)$\\  [0.5ex]  
\hline
 \end{tabular}
\end{center}
\end{table*}

\section{Analysis of Symmetric Replicated Erasure Codes }
\label{sec:symmetric}
In this section, we analyze the expected data persistency of replicated erasure codes for the \textbf{symmetric replication strategy.}
We note that \textit{the concept of the expected data persistency for symmetric replication strategy}
was investigated in Marchwicki's PhD thesis~\cite{phdkarol}. 
Namely, \textbf{symmetric replication of erasure codes} was analyzed
and the asymptotic result formulated in Theorem~\ref{thm:mainasymer} of the current paper was earlier obtained for very special parameters
$p\in\mathbb{N}\setminus\{0\},$ $q\in\mathbb{N}$  and $r=1$.
In our notation, this means for the case $REC(p,p+q,1)$.
Thus, in Theorem~\ref{thm:mainasymer} below we generalizes one of the main results of PhD thesis~\cite{phdkarol} for all parameters~$r\in\mathbb{N}\setminus\{0\}$.
We are now ready to present and analyse Algorithm~\ref{alg_symnew}.

\begin{algorithm}[H]
\caption{Symmetric replication of $REC(p,p+q,r).$} 
\label{alg_symnew}
\begin{algorithmic}[1]
 \REQUIRE A system of $N$ storage nodes; $f^{(1)}, f^{(2)},\dots f^{(D)}$ data items; each of the data items is encoded and replicated into $(p+q)r$ chunks (see Definition \ref{definition:first});
 $f^{(1)}_1(j),f^{(1)}_2(j),\dots f^{(1)}_{p+q}(j),$ $f^{(2)}_1(j),f^{(2)}_2(j),\dots f^{(2)}_{p+q}(j),\dots$ $f^{(d)}_1(j),f^{(d)}_2(j),\dots f^{(D)}_{p+q}(j),$
 for $j\in\{1,2\dots, r\},$
 fragments after encoding and replication.
 \ENSURE  The final positions of $(p+q)rD$ fragments in the nodes of the storage system (see Figure $3$).
 \STATE{$A:=1;$}
 \FOR{$k=1$  \TO $D$ } 
 \FOR{$j=1$  \TO $r$ }
 \FOR{$l=1$  \TO $p+q$ } 
 \STATE{replicate fragment $f^{(k)}_{l}(j)$ at the node with identifier $A;$}
 \STATE{$A:=A+1\mod N;$}
\ENDFOR
 \ENDFOR
 \ENDFOR
\end{algorithmic}
\end{algorithm}

Figure $3$ illustrates Algorithm~\ref{alg_symnew} for the  system parameters $p=q=1,$ $r=2,$ number of documents $D=7$ and number of nodes $N=16$.
In this case, each document is divided and encoded into $4$ chunks. 
Hence, we have $28$ chunks replicated according to the symmetric strategy in $16$ nodes.

\begin{figure*}[ht]
  $$
\begin{array}{|c|c|c|c|c|c|c|c|c|c|}
\hline
 n & 1 & 2 & 3  & 4 & 5 & 6 & 7 & 8  \\
 \hline
 &f^{(1)}_1(1) & f^{(1)}_2(1)  &  f^{(1)}_1(2)  & f^{(1)}_2(2)  & f^{(2)}_1(1) &  f^{(2)}_2(1)
  &  f^{(2)}_1(2) 
 & f^{(2)}_2(2)  \\
 
 &f^{(5)}_1(1) & f^{(5)}_2(1)  &  f^{(5)}_1(2)  & f^{(5)}_2(2)  & f^{(6)}_1(1) &  f^{(6)}_2(1)
  &  f^{(6)}_1(2) 
 & f^{(6)}_2(2)
 \\
 \hline
 \hline
 n & 9 & 10 & 11  & 12 & 13 & 14 & 15 & 16  \\
 \hline
 &f^{(3)}_1(1) & f^{(3)}_2(1)  &  f^{(3)}_1(2)  & f^{(3)}_2(2)  & f^{(4)}_1(1) &  f^{(4)}_2(1)
  &  f^{(4)}_1(2) 
 & f^{(4)}_2(2)  \\
 
  &f^{(7)}_1(1) & f^{(7)}_2(1)  &  f^{(7)}_1(2)  & f^{(7)}_2(2)  &  &  &  & 
 \\
 
 \hline
\end{array}
$$ 
  \caption{Symmetric replicated strategy according to Algorithm \ref{alg_symnew} for replicated erasure codes $REC(1,2,2),$ i.e. the 
  system parameters $p=q=1,$ $r=2,$ the number of documents $D=7$ and the number of nodes $n=16.$} 
\end{figure*}

We now prove the following exact formula for the expected data persistency of replicated erasure codes for the symmetric replication strategy as the integral of incomplete Beta functions.
It is worth pointing out that the proof of Theorem~\ref{thm:mainasymerfirst} is technically complicated. 
However, as in the proof of Corollary~\ref{thm:random_a} in the previous section,
we derive $\Pr\left[X_{(p,p+q,r)}^{(N,D)}>l\right]$ (see Equations~(\ref{eq:forKal}), (\ref{eq:basici}))
and apply Identity~(\ref{eq:random_variable}) to get the desired exact integral formula for $\E{K_{(p,p+q,r)}^{(N,D)}}$ (see Equations~(\ref{eq:dusa}), (\ref{eq:laster33}) and~(\ref{eq:laster33a})).

\begin{theorem} 
\label{thm:mainasymerfirst}
Let Assumption \ref{assumption:persistency} hold. 
Assume that all chunks are replicated according to the symmetric
strategy in the storage nodes as in Algorithm~\ref{alg_rand}. 
Let $(p+q)r$ be a divisor of the number of nodes $N$ and the number of documents $D\ge\frac{N}{(p+q)r}$.
Then
$$
\E{K_{(p,p+q,r)}^{(N,D)}}=
(N+1)\int_{0}^{1}\left(1-\left(I_{x}(q+1,p)\right)^r\right)^{\frac{N}{(p+q)r}}dx,
$$
where $I_{x}(q+1,p)$ is the incomplete Beta function.
\end{theorem}
We note that the expected data persistency of symmetric replicated erasure codes \textbf{does not depend} on the number of documents $D$.
\begin{proof} 
Fix $p\in\mathbb{N}\setminus\{0\},$ $q\in\mathbb{N}$ and $r\in\mathbb{N}\setminus\{0\}.$ 
Let $\Omega$ denote the set of nodes in the storage system. 
Assume that $|\Omega|=N$ and $\frac{N}{(p+q)r}\in\mathbb{N},$ and all chunks are replicated according to the symmetric
strategy in the storage nodes as in Algorithm~\ref{alg_rand}. 

We now make the following \textbf{important observation.}
Namely,
\begin{align*}
U_{i,j}=&\{(p+q)(r(i-1)+(j-1))+1,(p+q)(r(i-1)+(j-1))+2,\\
&\dots, (p+q)(r(i-1)+(j-1))+(p+q)\},
\end{align*}
for $i=1,2,\dots,\frac{n}{(p+q)r}$ and for  $j=1,2,\dots,r$
is the fixed partition of the set $\Omega$ into
disjoint clusters of cardinality $p+q$. 

Suppose that we are successively, independently 
and randomly removing distinct nodes from the set $\Omega$.
Let $\omega_1, \omega_2, \ldots$ be a realization of this process.
Observe that
$$
  X_{(p,p+q,r)}^{(N,D)} = \min\{l: (\exists i)(\forall j)(U_{i,j} \cap \{\omega_1,\ldots,\omega_l\}= q+1)\}
$$ 
(see Assumption \ref{assumption:persistency}).
Let $[\Omega]^l$ denote the family of all subsets
of $\Omega$ of cardinality $l$. Let $X$ be a random element from $[\Omega]^l.$ 
Then
\begin{equation}
  \label{eq:forKal}
  \Pr[X_{(p,p+q,r)}^{(N,D)}> l] = \frac{|\{X\in[\Omega]^l:(\forall i)(\exists j)(0 \le  U_{i,j} \cap X\le q)\}|}{|[\Omega]^l|} ~.
\end{equation} 

Let us consider the space $[\Omega]^l$ with uniform probability. 
Assume that 
$$(\forall i)(\exists j)(0 \le  U_{i,j} \cap X\le q).$$

Let $j_1,j_2,\dots,j_r\in\{0,1,\dots, p+q\}.$
Define
$$
A(p,q,r)
=\{(j_1,j_2,\dots, j_r)\in\{0,1,\dots p+q\}: j_1\le q \, \vee \,j_2\le q\vee\dots\vee j_r\le q\}.
$$
Let $k_{(j_1,j_2,\dots,j_r)}$ denote the number of
occurrences of  $(j_1,j_2,\dots,j_r)$ in the sequence
$$\{(U_{i,1}\cap X,U_{i,2}\cap X,\dots,U_{i,r}\cap X)\}_{i\in\left\{1,2,\dots,\frac{N}{(p+q)r}\right\}}.$$ 
Observe that
\begin{equation}
\label{eq:sumere01}
\sum_{A(p,q,r)} k_{(j_1,j_2,\dots,j_r)}=\frac{N}{(p+q)r},
\end{equation}
\begin{equation}
\label{eq:sumere02}
\sum_{A(p,q,r)} (j_1+j_2+\dots +j_r)k_{(j_1,j_2,\dots,j_r)}=l.
\end{equation}
\begin{align}
\label{eq:forkal03}D(p,q,r,l)=\{k_{(j_1,j_2,\dots,j_r)}: \sum_{A(p,q,r)} (j_1+j_2+\dots +j_r)k_{(j_1,j_2,\dots,j_r)}=l\}.
\end{align}
Putting all together (\ref{eq:forKal}-\ref{eq:forkal03}) we derive
\begin{align}
\nonumber&\Pr[X_{(p,p+q,r)}^{(N,D)}> l]\\ 
 \label{eq:basici}&=  \frac{\sum\limits_{D(p,q,r,l)}\left(\frac{\left(\frac{N}{(p+q)r}\right)!}{\prod\limits_{A(p,q,r)}\left(k_{(j_1,j_2,\dots,j_r)}\right)!}\right) 
 \prod\limits_{A(p,q,r)}\left(\binom{p+q}{j_1}\binom{p+q}{j_2}\dots\binom{p+q}{j_r}\right)^{k_{(j_1,j_2,\dots,j_r)}}}
    {\binom{N}{l}}.
\end{align}
Observe that $\Pr[X_{(p,p+q,r)}^{(N,D)}> l]=0$ for any $l\ge q\frac{N}{p+q}+1$ (see Equation (\ref{eq:sumere01}) and Equation (\ref{eq:sumere02})
for $j_1=j_2=\dots =j_q=r$).
Using this and Identity (\ref{eq:random_variable})
for random variable $X:=X_{(p,p+q,r)}^{(N,D)}$ we have
\begin{equation}
\label{eq:cool}
\E{X_{(p,p+q,r)}^{(N,D)}}=\sum_{l\ge 0}\Pr[X_{(p,p+q,r)}^{(N,D)}> l]=\sum_{l=0}^{q\frac{N}{p+q}}\Pr[X_{(p,p+q,r)}^{(N,D)}> l].
\end{equation}
Combining together (\ref{eq:basici}) and (\ref{eq:cool}) we derive \textbf{the first complicated formula} for $\E{X_{(p,p+q,r)}^{(N,D)}}.$
\begin{align}
\nonumber&\E{X_{(p,p+q,r)}^{(N,D)}}=\sum_{l=0}^{q\frac{N}{p+q}}{\binom{N}{l}}^{-1}\sum\limits_{D(p,q,r,l)}\left(\frac{\left(\frac{N}{(p+q)r}\right)!}{\prod\limits_{A(p,q,r)}\left(k_{(j_1,j_2,\dots,j_r)}\right)!}\right)\times\\ 
 \label{eq:dusa}&\times\prod\limits_{A(p,q,r)}\left(\binom{p+q}{j_1}\binom{p+q}{j_2}\dots\binom{p+q}{j_r}\right)^{k_{(j_1,j_2,\dots,j_r)}}.   
\end{align}
The rest of this proof is devoted to \textbf{simplification of expression} (\ref{eq:dusa}) for  $\E{X_{(p,p+q,r)}^{(N,D)}}.$
Using Formula (\ref{eq:sumere02}) we get
\begin{equation}
\label{eq:sumere03}
\left(\frac{x}{1-x}\right)^l=\prod\limits_{A(p,q,r)}\left(\frac{x}{1-x}\right)^{(j_1+j_2+\dots+j_r)k_{(j_1,j_2,\dots,j_r)}}.
\end{equation}
Applying Identities (\ref{eq:first01}), (\ref{eq:first02}) for $a:=l+1$ and $b:=N+1-l$ we have
$$
{\binom{N}{l}}^{-1}=(N+1)\int_{0}^{1}x^l(1-x)^{N-l}dx.
$$
This together with (\ref{eq:dusa}) and (\ref{eq:sumere03}), as well as interchanging summation with integral lead to 
\begin{align*}
&\E{X_{(p,p+q,r)}^{(N,D)}}=(N+1)\int_{0}^{1} (1-x)^{N}\sum_{l=0}^{q\frac{N}{p+q}}\sum\limits_{D(p,q,r,l)}\left(\frac{\left(\frac{N}{(p+q)r}\right)!}{\prod\limits_{A(p,q,r)}\left(k_{(j_1,j_2,\dots,j_r)}\right)!}\right)\times\\ 
&\times\prod\limits_{A(p,q,r)}\left(\binom{p+q}{j_1}\binom{p+q}{j_2}\dots\binom{p+q}{j_r}\left(\frac{x}{1-x}\right)^{j_1+j_2+\dots+j_r}\right)^{k_{(j_1,j_2,\dots,j_r)}}dx. 
\end{align*}
From Equation (\ref{eq:sumere01}) and multinomial theorem we get
\begin{align}
\nonumber &\E{X_{(p,p+q,r)}^{(N,D)}}=(N+1)\int_{0}^{1} (1-x)^{N}\times\\ 
\nonumber&\times \left(\sum\limits_{A(p,q,r)}\binom{p+q}{j_1}\binom{p+q}{j_2}\dots\binom{p+q}{j_r}\left(\frac{x}{1-x}\right)^{j_1+j_2+\dots+j_r}\right)^{\frac{N}{(p+q)r}}dx\\
\label{eq:laster33}&=(N+1)\int_{0}^{1}\left(\sum\limits_{A(p,q,r)}\prod_{k=1}^r\binom{p+q}{j_k}\left(x\right)^{j_k}\left(1-x\right)^{p+q-j_k}\right)^{\frac{N}{(p+q)r}}dx.
\end{align}
We now define
\begin{align*}
&B(p,q,r)=\\
&\{(j_1,j_2,\dots, j_r)\in\{0,1,\dots p+q\}: j_1\ge q+1 \, \wedge \,j_2\ge q+1\wedge\dots\wedge j_r\ge q+1\}.
\end{align*}
Observe that
\begin{equation}
\label{eq:suma01}
A(p,q,r)\cup B(p,q,r)=\{(j_1,j_2,\dots, j_r)\in\{0,1,\dots p+q\}\},
\end{equation}
\begin{equation}
\label{eq:suma02}
A(p,q,r)\cap B(p,q,r)=\emptyset.
\end{equation}
Combining together (\ref{eq:suma01}) and the binomial theorem we derive
\begin{align*}
\sum\limits_{A(p,q,r)\cup B(p,q,r)}&\prod_{k=1}^r\binom{p+q}{j_k}\left(x\right)^{j_k}\left(1-x\right)^{p+q-j_k}\\
&=\prod_{k=1}^r\sum_{j_k=0}^{p+q}\binom{p+q}{j_k}\left(x\right)^{j_k}\left(1-x\right)^{p+q-j_k}=1.
\end{align*}
Using this, Identity (\ref{eq:suma02}), the binomial theorem, as well as  Formula (\ref{eq:incomplete02}) for $a:=q+1,$ $b:=p$ we derive
\begin{align}
\nonumber\sum\limits_{A(p,q,r)}&\prod_{k=1}^r\binom{p+q}{j_k}\left(x\right)^{j_k}\left(1-x\right)^{p+q-j_k}\\
\nonumber&=1-\sum\limits_{B(p,q,r)}\prod_{k=1}^r\binom{p+q}{j_k}\left(x\right)^{j_k}\left(1-x\right)^{p+q-j_k}\\
\label{eq:last555}&=1-\left(\sum_{j=q+1}^{p+q}x^j(1-x)^{p+q-j}\right)^r=1-\left(I_{x}(q+1,p)\right)^r.
\end{align}
Finally, combining together (\ref{eq:laster33}) and (\ref{eq:last555}) we have the desired Formula
\begin{equation}
 \label{eq:laster33a}
\E{K_{(p,p+q,r)}^{(N,D)}}=
(N+1)\int_{0}^{1}\left(1-\left(I_{x}(q+1,p)\right)^r\right)^{\frac{N}{(p+q)r}}dx.
\end{equation}
This completes the proof of Theorem \ref{thm:mainasymerfirst}.
\end{proof}
We now observe that like for the random replication strategy, the maximal expected data persistency of replicated erasure codes 
for symmetric replication strategy
is attained by the replicated erasure codes with parameter $p=1.$
\begin{corollary}
\label{thm:mainbdefabc}
Under the assumption of Theorem \ref{thm:mainasymerfirst} we have
$$
 \max_{p\ge 1} \E{X_{(p,p+q,r)}^{(N,D)}} =\E{X_{(1,1+q,r)}^{(N,D)}}.
$$
\end{corollary}
\begin{proof} 
The proof of Corollary \ref{thm:mainbdefabc} is analogous to that of Corollary \ref{thm:mainbdef}.
Using (\ref{align:b100}) 
we have
$$
I_{x}(q+1,p+1)\ge I_{x}(q+1,p).
$$
Hence
$$
\left(1-\left(I_{x}(q+1,p)\right)^r\right)^{\frac{N}{(p+q)r}}\ge
\left(1-\left(I_{x}(q+1,p+1)\right)^r\right)^{\frac{N}{(p+q)r}}.
$$
Applying Theorem \ref{thm:mainasymerfirst} we get
$$\E{X_{(p,p+q,r)}^{(N,D)}}\ge \E{X_{(p+1,p+1+q,r)}^{(N,D)}}.$$
Therefore 
$$
 \max_{p\ge 1} \E{X_{(p,p+q,r)}^{(N,D)}} =\E{X_{(1,1+q,r)}^{(N,D)}}.
$$
This completes the proof of Corollary \ref{thm:mainbdefabc}.
\end{proof}
As before, the formula in Theorem~\ref{thm:mainasymerfirst} makes it hard to apply in practice.
To that end, we formulate in the next theorem a very useful asymptotic approximation for $\E{X_{(p,p+q,r)}^{(N,D)}}$. 
Particularly, we analyze the expected data persistency of replicated erasure codes for the symmetric replication strategy when the number of nodes $N$ is large. 
In the proof of Theorem~\ref{thm:mainasymer} 
we apply again the Erd\'elyi's formulation of Laplace's method for the integral formula obtained in Theorem~\ref{thm:mainbdefabc}.
%
\begin{theorem} 
\label{thm:mainasymer}
Let Assumption~\ref{assumption:persistency} hold. 
Assume that all chunks are replicated according to the symmetric
strategy in the storage nodes as in  Algorithm~\ref{alg_rand}. 
Let $(p+q)r$ be a divisor of the number of nodes $N$ and the number of documents $D\ge\frac{N}{(p+q)r}$.
Then
\begin{equation}
 \E{X_{(p,p+q,r)}^{(N,D)}} =\frac{\Gamma\left(1+\frac{1}{r(q+1)}\right)\left((p+q)r\right)^{\frac{1}{r(q+1)}}}{\binom{p+q}{q+1}^{\frac{1}{q+1}}}N^{1-\frac{1}{r(q+1)}}
  +O\left(\frac{1}{N^{\frac{2}{r(q+1)}}}\right),
\end{equation}
where $\Gamma\left(1+\frac{1}{r(q+1)}\right)$ is the Gamma function.
\end{theorem}
\begin{proof} 
The proof of Theorem~\ref{thm:mainasymer} is analogous to that of Theorem~\ref{thm:mainbd}.
Assume that $p\in\mathbb{N}\setminus\{0\},$  $q\in\mathbb{N}$ and $r\in\mathbb{N}\setminus\{0\}$ are fixed and independent on $N$.
Let 
\begin{equation}
\label{eq:lestacek}
J(N)=\int_{0}^{1}\left(1-\left(I_{x}(q+1,p)\right)^r\right)^{\frac{N}{(p+q)r}}dx.
\end{equation}
In the asymptotic analysis of $J(N)$ for large $N$ we apply the Erd\'elyi's formulation of Laplace's method. 
Observe that
$$J(N)=\int_{0}^{1}e^{-\frac{N}{(p+q)r}f(x)}dx,$$
where $f(x)=-((p+q)r)^{-1}\ln\left(1-\left(I_{x}(q+1,p)\right)^r\right)$.

From Lemma~\ref{lemma:pochodna1}, we deduce that the assumption of Theorem~1.1 in~\cite{Nemes2013} holds for $f(x):=-((p+q)r)^{-1}\ln\left(1-\left(I_{x}(q+1,p)\right)^r\right),$ $g(x):=1,$  $a:=0$ and $b:=1$.
Namely, we  apply Formula~$(1.5)$ in~\cite[Theorem 1.1]{Nemes2013} 
for  $a:=0$ $f(0):=0,$ $\lambda:=N,$ $\beta:=1,$ $\alpha:=r(q+1),$  $b_0=1,$  $a_0=((p+q)r)^{-1}\binom{p+q}{q+1}^r,$ $n:=k$   and deduce that the integral $J(N)$ has the following asymptotic expansion
$$\sum_{k=0}^{\infty}\Gamma\left(\frac{k+1}{r(q+1)}\right)\frac{c_k}{\left(N\right)^{\frac{k+1}{r(q+1)}}},$$
as $N\rightarrow\infty,$ where $(c_0)^{-1}=\left(\frac{1}{(p+q)r}\right)^{\frac{1}{r(q+1)}}\binom{p+q}{q+1}^{\frac{1}{q+1}}(r(q+1))$ and $\Gamma\left(\frac{k+1}{r(q+1)}\right)$ is the Gamma function.
Therefore
$$J(N)=\frac{\Gamma\left(\frac{1}{r(q+1)}\right)\frac{1}{r(q+1)}}{{\binom{p+q}{q+1}}^{\frac{1}{q+1}}}\frac{\left((p+q)r\right)^{\frac{1}{r(q+1)}}}{N^{\frac{1}{r(q+1)}}}+O\left(\frac{1}{N^{\frac{2}{r(q+1)}}}\right).$$
Using the basic identity for the Gamma function 
$$\Gamma\left(\frac{1}{r(q+1)}\right)\frac{1}{r(q+1)}=\Gamma\left(1+\frac{1}{r(q+1)}\right)$$ 
(see \cite[Identity 5.5.1]{NIST} for $z:=\frac{1}{r(q+1)}$)
we get 
\begin{equation}
\label{eq:lesta01abc}
J(N)=\frac{\Gamma\left(1+\frac{1}{r(q+1)}\right)}{{\binom{p+q}{q+1}}^{\frac{1}{q+1}}}\frac{\left((p+q)r\right)^{\frac{1}{r(q+1)}}}{N^{\frac{1}{r(q+1)}}}+O\left(\frac{1}{N^{\frac{2}{r(q+1)}}}\right).
\end{equation}
Finally, combining together the result of Theorem~\ref{thm:mainasymerfirst},~(\ref{eq:lestacek}) and~(\ref{eq:lesta01abc}), we have 
$$
  \E{X_{(p,p+q,r)}^{(N,D)}} =\frac{\Gamma\left(1+\frac{1}{r(q+1)}\right)\left((p+q)r\right)^{\frac{1}{r(q+1)}}}{\binom{p+q}{q+1}^{\frac{1}{q+1}}}N^{1-\frac{1}{r(q+1)}}
  +O\left(\frac{1}{N^{\frac{2}{r(q+1)}}}\right).
$$
This completes the proof of Theorem \ref{thm:mainasymer}. 
\end{proof}
Table~\ref{tab:mains} explains Theorem~\ref{thm:mainasymer} for $p\in\mathbb{N}\setminus\{0\}$ and some fixed parameters $q,r$ and in terms of asymptotic notations. 
In view of Theorem~\ref{thm:mainasymer}, we have $$\E{X_{(p,p+q,r)}^{(N,D)}}=\Theta\left(N^{1-\frac{1}{r(q+1)}}\right).$$
\begin {table*}[!ht]
\caption {The expected data persistency of replicated erasure codes $REC(p,p+q,r)$ for \textbf{symmetric replication strategy} as a function of the system parameters $p,q,r,$ provided that $N$ is large.}
\label{tab:lines} 
\begin{center}
 \label{tab:mains}
 \begin{tabular}{|*{4}{c|}} 
 \hline
 parameter $p$  & parameter $q$ & parameter $r$  & $\E{X_{(p,p+q,r)}^{(N,D)}}$ \\ [0.5ex]
    \hline
 $p\in\mathbb{N}\setminus\{0\}$ & $q=0$  & $r=1$ &    ${\Theta\left(1\right)}$ \\ [0.5ex]
 \hline
$p\in\mathbb{N}\setminus\{0\}$  & $q=0$ & $r=2$ &  ${\Theta\left(N^{\frac{1}{2}}\right)}$ \\  [0.5ex]
 \hline
$p\in\mathbb{N}\setminus\{0\}$  &  $q=1$ & $r=2$ &  ${\Theta\left(N^{\frac{3}{4}}\right)}$ \\  [0.5ex]
\hline
$p\in\mathbb{N}\setminus\{0\}$  & $q=2$ & $r=3$ &   ${\Theta\left(N^{\frac{8}{9}}\right)}$\\  [0.5ex]  
\hline
 \end{tabular}
\end{center}
\end{table*}

\section{Discussions}
\label{sec:discussions}
In this study, we analyzed the expected data persistency of $REC(p,p+q,r)$ (replicated erasure codes) in distributed storage systems for the random and symmetric replication strategies.
Recall that in $REC(p,p+q,r)$, each document is encoded and replicated.
To restore the original document, one has to download $p$ chunks, each from $r$ different replication sets (see Assumption~\ref{assumption:persistency}).
Thus, if $q+1$ chunks in each $r$ replication sets are placed in the nodes that leave the storage system, then restoring the original document becomes impossible. 

The expected data persistency of a storage system for the random (as well as the symmetric) strategy depends on $r(q+1)$. 
Namely, the expected data persistency is in $N\Theta\left(D^{-\frac{1}{r(q+1)}}\right)$ for the random replication strategy and
is in $\Theta\left(N^{1-\frac{1}{r(q+1)}}\right)$ for the symmetric replication strategy, where $N$ is the number of nodes in the storage system and $D$ is the number of documents.
Table~\ref{tab:dwasemen} displays a comparison of the strategies: random and symmetric.

\begin{table}[H]
\caption {Comparision of the expected data persistency for random and symmetric erasure codes.} \label{tab:hugher} 
\begin{center}
\label{tab:dwasemen}
 \begin{tabular}{|c|c|c|c|} 
 \hline
 \begin{tabular}{c}Number \\ of documents\end{tabular}  & 
Replication strategy & \begin{tabular}{c}Expected \\ data persistency\end{tabular} \\
  \hline
 $D=O(N)$ & random & $N\Theta\left(D^{-\frac{1}{r(q+1)}}\right)$ \\ 
 \hline
 $D=\Theta(N)$ & random or symmetric & $\Theta\left(N^{1-\frac{1}{r(q+1)}}\right)$ \\ 
 \hline
 $D=\Omega(N)$ & symmetric & 
$\Theta\left(N^{1-\frac{1}{r(q+1)}}\right)$\\ 
 \hline
 \end{tabular}
\end{center}
\end{table}

Let us consider case $D=\Delta N,$ provided that $\Delta$ is constant, i.e., $\Delta$ is the average number of documents per node. 
Then both strategies give asymptotically the same expected data persistency.
When $D=O(N)$, then the random strategy gives highest expected data persistency.
If $D=\Omega(N)$, then the symmetric strategy is better.

\subsection{Non-uniform Redundancy Scheme}
In this paper, it is assumed that each document is encoded and replicated according to $REC(p,p+q,r)$ for some fixed parameters $p,q,r$.
However, our analysis is not limited to \textit{uniform} settings only, and the proposed
theory can also handle \textit{non-uniform} redundancy schemes, where each document is encoded and replicated according to different parameters~$p,q,r$.

Namely, assume that $D_i$ is the number of documents that are encoded and replicated according to $REC(p_i,p_i+q_i,r_i)$
for some fixed $p_i\in\mathbb{N}\setminus\{0\},$ $q_i\in\mathbb{N}$ and $r_i\in\mathbb{N}\setminus\{0\}$, 
provided that $i=1,2,\dots ,k$ and $D_1+D_2+\dots D_k=D$.

Then, the  the expected data persistency of storage system with $D_i$ documents is in $N\Theta\left(D_i^{-\frac{1}{r(q+1)}}\right)$ for the random replication strategy and
is in $\Theta\left(N^{1-\frac{1}{r(q+1)}}\right)$ for the symmetric replication strategy, where $N$ is the number of nodes. 
Of course, the closed analytical and asymptotic formulas for expected data persistency of a storage system with $D_1+D_2+\dots D_k$ documents need further theoretical studies.

\subsection{Exact Formulas in Terms of Beta Functions}
\label{sub:exactly}
We proved that the maximal expected data persistency of $REC(p,p+q,r)$ for both replication strategies is attained by the replicated erasure codes with parameter $p=1$ (see Theorem~\ref{thm:mainbdef} in Section~\ref{section:random} and Theorem~\ref{thm:mainbdefabc} in Section~\ref{sec:symmetric}). 
Let us recall that for $p=1$ replicated erasure codes becomes
simply replication with $(1+q)r$ replicas. If we restrict to $p=1$, we can give exacts formulas for the expected data persistency of $REC(1,1+q,r)$ in terms of the Beta function (see Formula~(\ref{eq:beta777}) and Formula~(\ref{eq:beta888})).

\paragraph*{a) The random strategy}
Combining together Theorem~\ref{thm:mainbex} with Equation~(\ref{eq:incomplete02}) for $a:=q+1, b:=1$ and $x:=x^r$ we have
$$
\E{X_{(1,1+q,r)}^{(N,D)}} =N\int_{0}^{1}\left(\sum_{j=0}^{q}\binom{q+1}{j}x^{rj}\left(1-x^r\right)^{q+1-j}\right)^{D} dx+ER.
$$
By the binomial theorem we get
$$
\E{X_{(1,1+q,r)}^{(N,D)}} =N\int_{0}^{1}\left(1-x^{r(q+1)}\right)^{D} dx+ER.
$$
Together, substitution $y:=x^{r(q+1)}$ in the previous formula and Definition~\ref{eq:first01} for $a:=D+1,$ $b:=\frac{1}{r(q+1)}$ lead to
\begin{equation}
\label{eq:beta777}
\E{X_{(1,1+q,r)}^{(N,D)}}=\frac{N}{r(q+1)}\mathrm{Beta}\left(D+1,\frac{1}{r(q+1)}\right)+ER,\,\,\,\text{where}\,\,\,|ER|\le 1.
\end{equation}

\paragraph*{b) The symmetric strategy}
Combining together Theorem~\ref{thm:mainasymerfirst} 
with Equation~(\ref{eq:incomplete02}) for $a:=q+1, b:=1$  we have
$$
\E{X_{(1,1+q,r)}^{(N,D)}} =(N+1)\int_{0}^{1}\left(1-\left(1-\sum_{j=0}^{q}\binom{q+1}{j}x^{j}\left(1-x\right)^{q+1-j}\right)^r\right)^{\frac{N}{r(q+1)}}dx
$$
By the binomial theorem we get
$$
\E{X_{(1,1+q,r)}^{(N,D)}} =(N+1)\int_{0}^{1}\left(1-x^{(q+1)r}\right)^{\frac{N}{r(q+1)}} dx.
$$
Together substitution $y:=x^{r(q+1)}$ in the previous formula and Definition~\ref{eq:first01} for $a:=\frac{N}{(q+1)r}+1,$ $b:=\frac{1}{r(q+1)}$ lead to
\begin{equation}
\label{eq:beta888}
\E{X_{(1,1+q,r)}^{(N,D)}}=\frac{N+1}{r(q+1)}\mathrm{Beta}\left(\frac{N}{r(q+1)}+1,\frac{1}{r(q+1)}\right).
\end{equation}
We note that the Exact Formula (\ref{eq:beta888}) was obtained in
\cite{JCIChordSSS} for the very special parameter $r=1.$

\section{Numerical Results}
\label{sec:experiments}
In this section, we use a set of experiments to show how the parameters $p,q,r$ impact the expected data persistency of the two replication strategies: random and symmetric. 
Algorithms~\ref{alg_rand} and~\ref{alg_symnew} are implemented to illustrate Theorem~\ref{thm:mainbd} together with the exact Formula~(\ref{eq:beta777})
and Theorem~\ref{thm:mainasymer}. 

We repeated the following experiments $10$ times. 
First, for each number of nodes $N\in\{48*1, 48*2, 48*3, 48*62\}$, we generated $50$ replicas of $D$ documents according to replicated erasure codes and Algorithms~\ref{alg_rand} and~\ref{alg_symnew} respectively. 
Let $E_{n,50}$ be the average of $50$ measurements of the data persistency.
Then, we placed the points in the set $\{(n,E_{n,50}):n=48*1,48*2,48*3,\dots,48*62\}$ into the picture.

\begin{minipage}{\linewidth}
      \centering
      \begin{minipage}{0.43\linewidth}
          \begin{figure}[H]
              \includegraphics[width=1.0\linewidth]{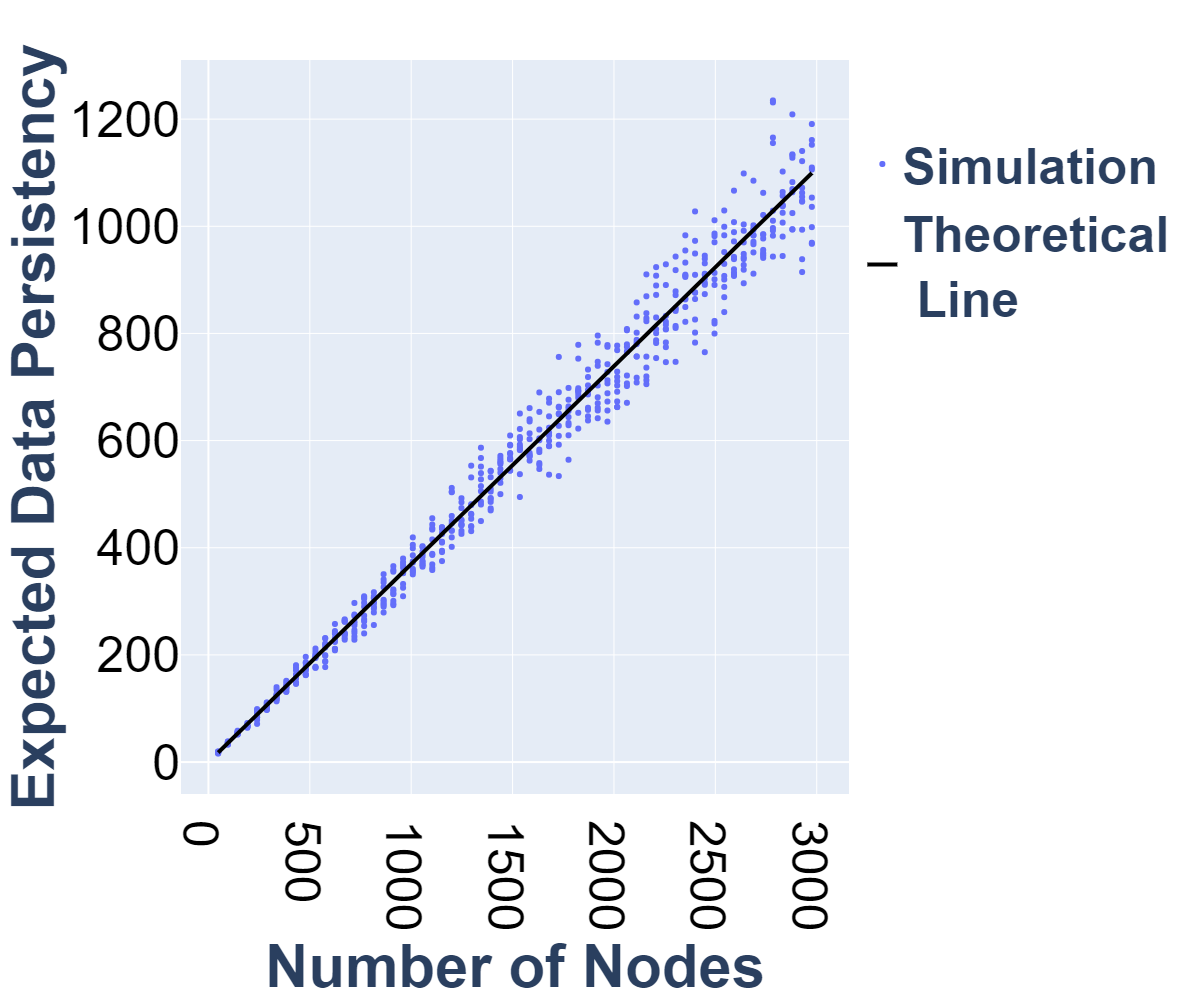}
              \vspace{2pt}
              \centerline{$\mathbf{E}\left[X_{(1,1,2)}^{(N,5)}\right]=\frac{\mathrm{Beta}\left(6,\frac{1}{2}\right)}{2}N$}
              \vspace{-10pt}
              \caption{The expected data persistency of Algorithm \ref{alg_rand} for $p=1,$ $q=0,$ $r=2$ and $D=5$}
              \label{fig:4}
          \end{figure}
      \end{minipage}
      \hspace{0.05\linewidth}
      \begin{minipage}{0.43\linewidth}
          \begin{figure}[H]
              \includegraphics[width=1.03\linewidth]{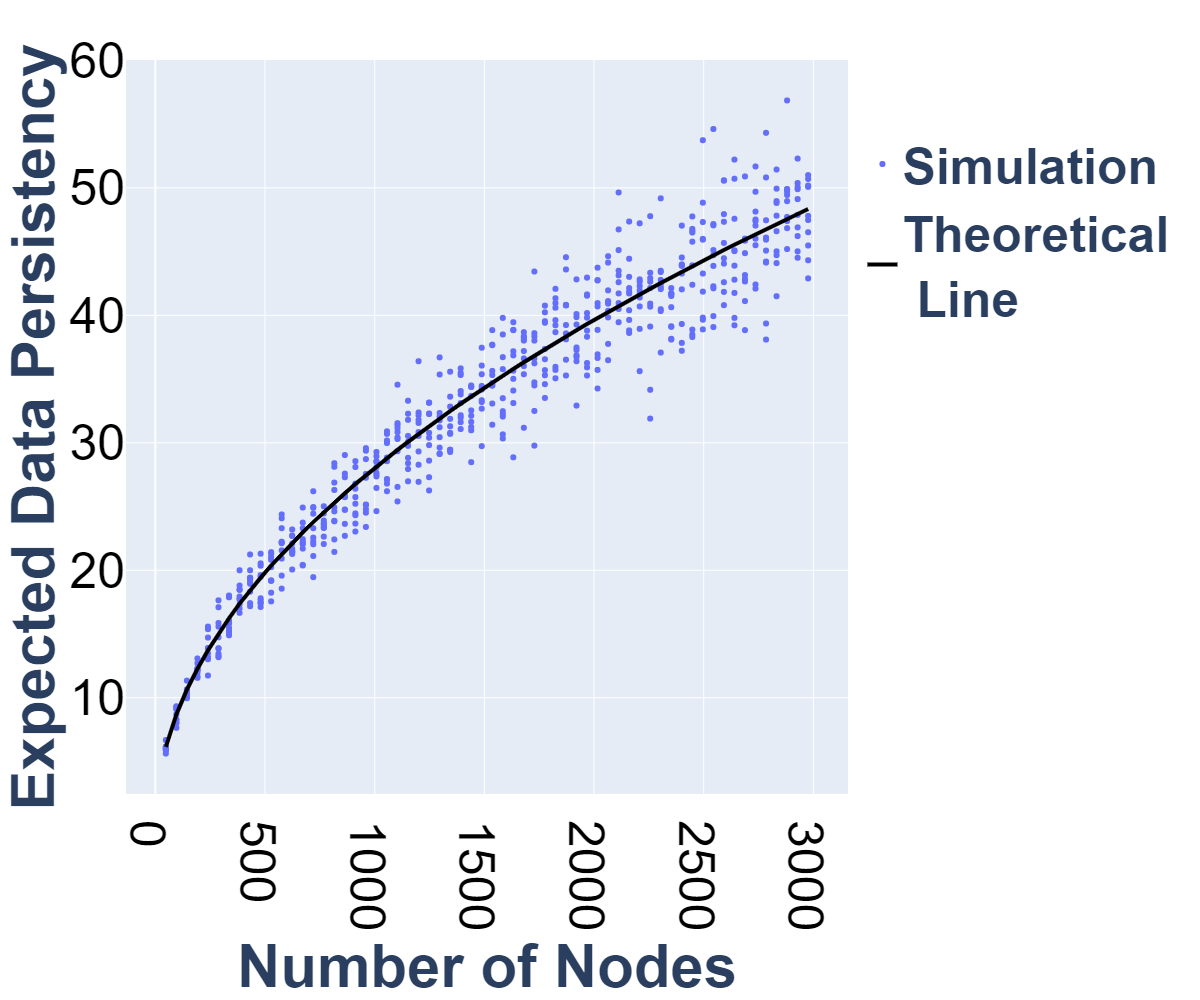}
               \vspace{2pt}
              \centerline{$\mathbf{E}\left[X_{(1,1,2)}^{(N,N)}\right]=\Gamma\left(\frac{3}{2}\right)\sqrt{N}$}
              \vspace{-10pt}
              \caption{The expected data persistency of Algorithm \ref{alg_rand} for $p=1,$ $q=0,$ $r=2$ and $D=N$}
              \label{fig:5}
          \end{figure}
      \end{minipage}
      \vspace{5pt}
  \end{minipage}

In Figures~\ref{fig:4} and~\ref{fig:5} the dots represent the experimental data persistency of Algorithm~\ref{alg_rand} considering the parameters $p=1,$ $q=0,$ $r=2,$ $D=5$
and $p=1,$ $q=0,$ $r=2,$ $D=N$ for the number of nodes $N\in\{48*1, 48*2, 48*3,\dots, 48*62\}$.
The additional lines $\left\{\left(N,\frac{\mathrm{Beta}\left(6,\frac{1}{2}\right)}{2}N\right), 48\le N\le 2976\right\},$\\
$\left\{\left(N,\Gamma\left(\frac{3}{2}\right)\sqrt{N}\right), 48\le N\le 2976\right\},$
are the leading terms in the theoretical estimation of the expected data persistency of Algorithm~\ref{alg_rand} when $p=1,$ $q=0,$ $r=2,$ $D=5$ 
(see exact Formula~(\ref{eq:beta777}) for $p=1,$ $q=0,$ $r=2,$ $D=5$)
and when $p=1,$ $q=0,$ $r=2,$ $D=N$ (see Theorem~\ref{thm:mainbd} for $p=1,$ $q=0,$ $r=2,$ $D=N$).
Figures~\ref{fig:4} and~\ref{fig:5} together illustrate the decline from $\frac{\mathrm{Beta}\left(6,\frac{1}{2}\right)}{2}N$ to 
$\Gamma\left(\frac{3}{2}\right)\sqrt{N}$ in the expected data persistency for the random replication strategy in Algorithm~\ref{alg_rand} when the number of documents $D$ increases from $5$ to $N$.

 \begin{minipage}{\linewidth}
      \centering
      \begin{minipage}{0.43\linewidth}
          \begin{figure}[H]
              \includegraphics[width=0.99\linewidth]{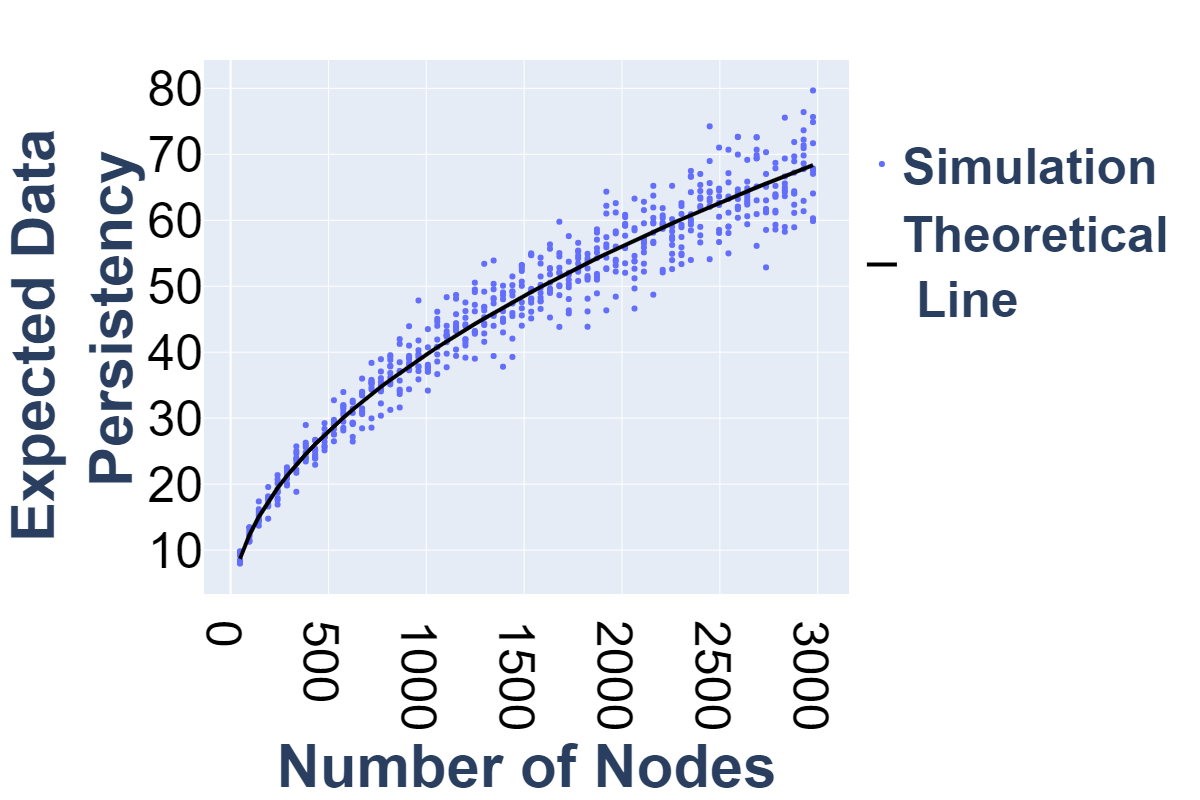}
              \vspace{2pt}
              \centerline{$\mathbf{E}\left[X_{(1,2,1)}^{(N,D)}\right]=\Gamma\left(\frac{3}{2}\right)2^{\frac{1}{2}}\sqrt{N}$}
              \vspace{-10pt}
              \caption{The expected data persistency of Algorithm \ref{alg_symnew} for $p=1,$ $q=1,$ $r=1$}
              \label{fig:6}
          \end{figure}
      \end{minipage}
      \hspace{0.05\linewidth}
      \begin{minipage}{0.43\linewidth}
          \begin{figure}[H]
              \includegraphics[width=1.01\linewidth]{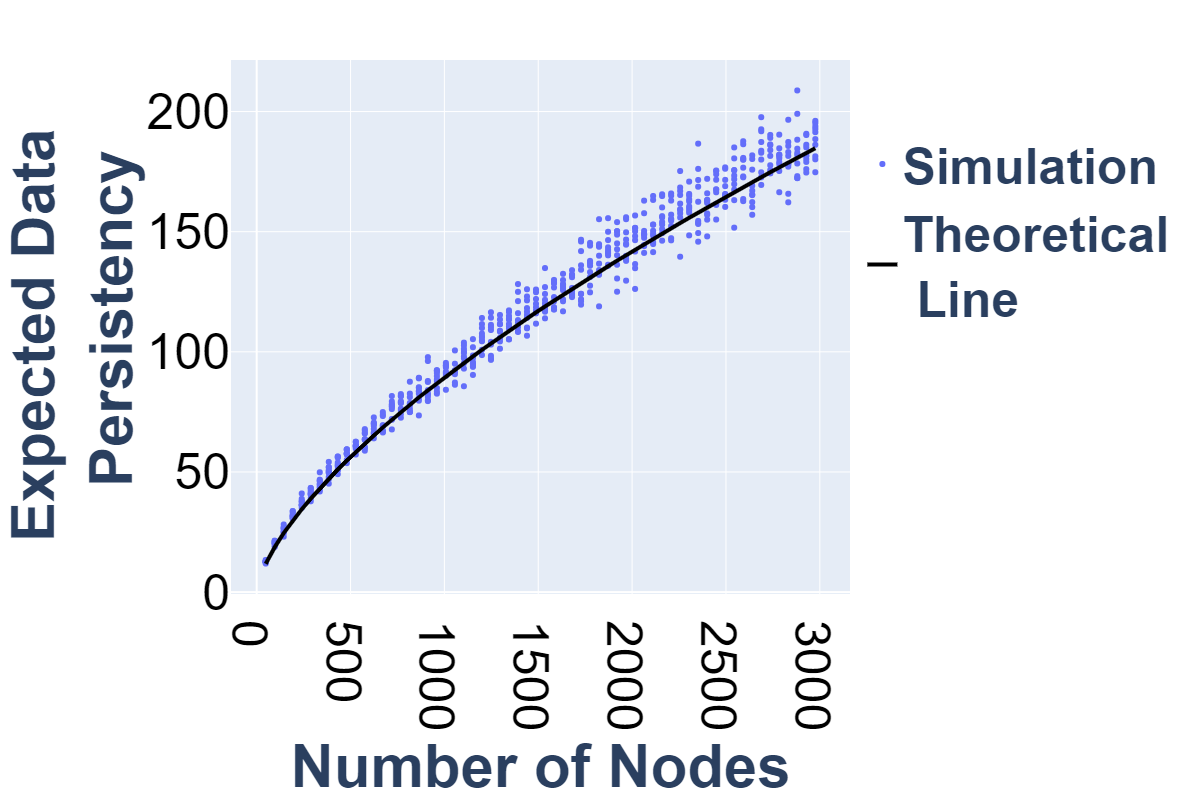}
               \vspace{2pt}
              \centerline{$\mathbf{E}\left[X_{(2,4,1)}^{(N,D)}\right]=\Gamma\left(\frac{4}{3}\right) N^{\frac{2}{3}}$}
              \vspace{-10pt}
              \caption{The expected data persistency of Algorithm \ref{alg_symnew} for $p=2,$ $q=2,$ $r=1$}
              \label{fig:7}
          \end{figure}
      \end{minipage}
      \vspace{5pt}
  \end{minipage}  

The experimental data persistency of Algorithm~\ref{alg_symnew} considering the parameters $p=1,$ $q=1,$ $r=1$ and $p=2,$ $q=2,$ $r=1$ is illustrated in Figures~\ref{fig:6} and~\ref{fig:7} for the number of nodes $N\in\{48*1, 48*2, 48*3,\dots, 48*62\}$. 
Similarly, the additional lines 
 $\left\{\left(N,\Gamma\left(\frac{3}{2}\right)2^{\frac{1}{2}}\sqrt{N}\right), 48\le N\le 2976\right\},$
$\left\{\left(N,\Gamma\left(\frac{4}{3}\right) N^{\frac{2}{3}}\right), 48\le N\le 2976\right\},$ are\\
the leading terms in the theoretical estimation of the expected data persistency of Algorithm~\ref{alg_symnew} when $p=1,$ $q=1,$ $r=1$
and when $p=2,$ $q=2,$ $r=1$  (see Theorem~\ref{thm:mainasymer} for $p=1,$ $q=0,$ $r=1$ and $p=2,$ $q=2,$ $r=1$).
Figures~\ref{fig:6} and~\ref{fig:7} together illustrate the increase from $\Gamma\left(\frac{3}{2}\right)2^{\frac{1}{2}}\sqrt{N}$ to 
$\Gamma\left(\frac{4}{3}\right) N^{\frac{2}{3}}$ in the expected data persistency for the random replication strategy in Algorithm~\ref{alg_symnew} when $r=1$ and $p,q$ increase from $1$ to $2$.

\begin{minipage}{\linewidth}
      \centering
      \begin{minipage}{0.43\linewidth}
          \begin{figure}[H]
              \includegraphics[width=0.8\linewidth]{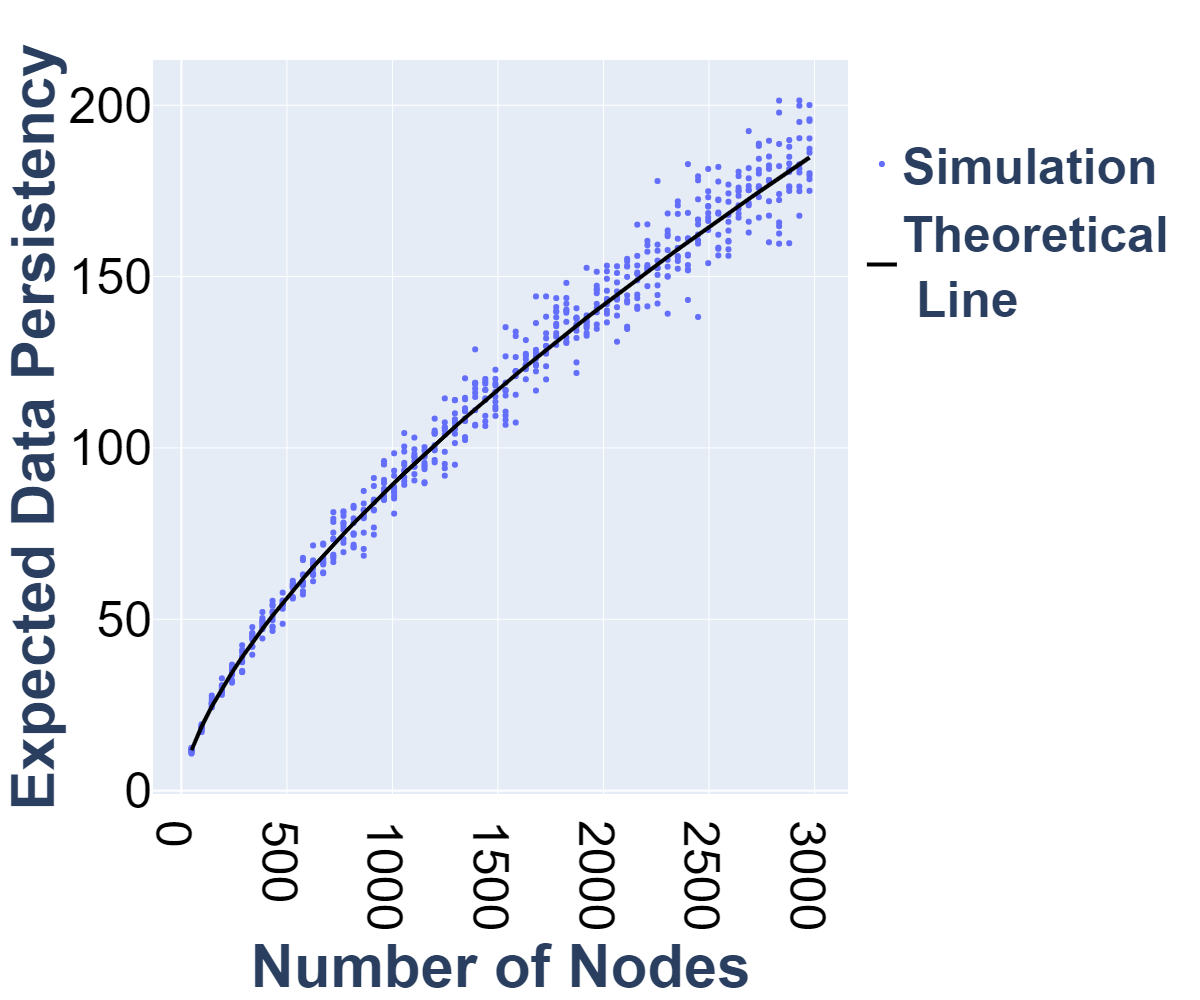}
              \vspace{2pt}
              \centerline{$\mathbf{E}\left[X_{(1,3,1)}^{(N,N)}\right]=\Gamma\left(\frac{4}{3}\right)N^{\frac{2}{3}}$}
              \vspace{-10pt}
              \caption{The expected data persistency of Algorithm \ref{alg_rand} for $p=1,$ $q=2,$ $r=1$ abd $D=N$}
              \label{fig:8}
          \end{figure}
      \end{minipage}
      \hspace{0.05\linewidth}
      \begin{minipage}{0.43\linewidth}
          \begin{figure}[H]
              \includegraphics[width=1.01\linewidth]{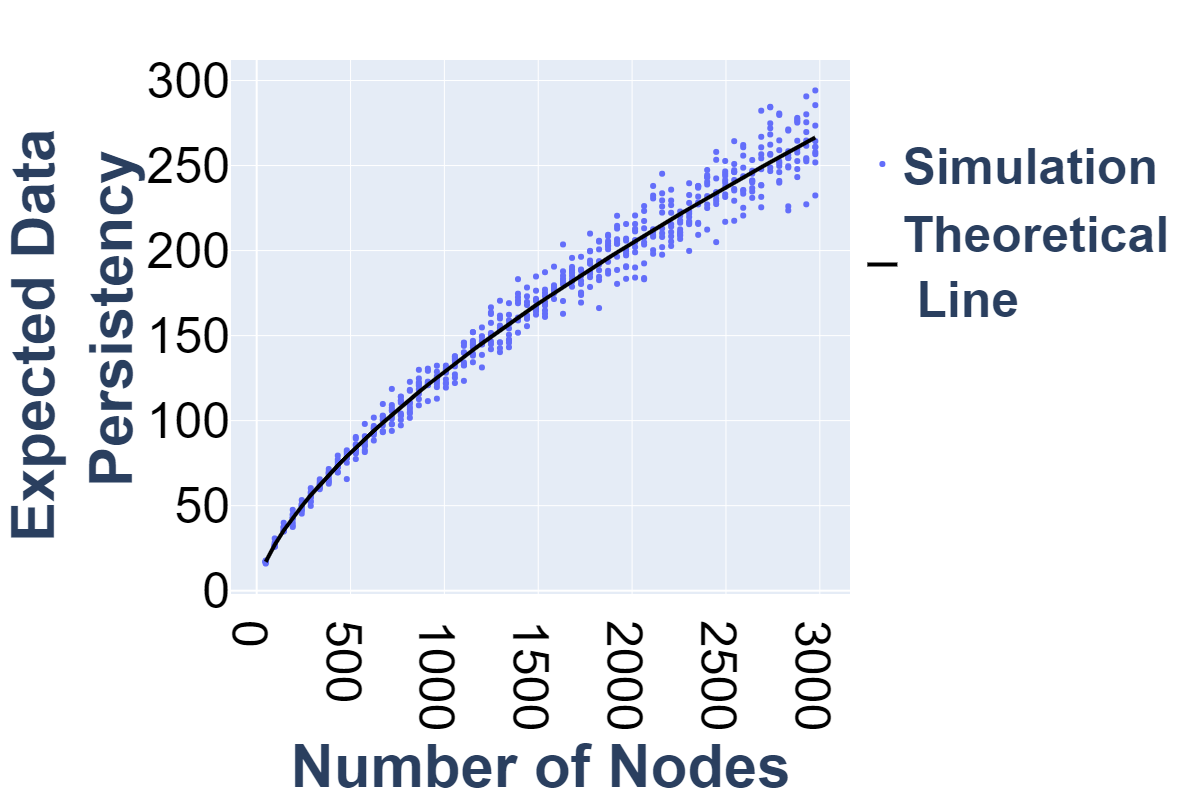}
               \vspace{2pt}
              \centerline{$\mathbf{E}\left[X_{(1,3,1)}^{(N,D)}\right]=\Gamma\left(\frac{4}{3}\right)3^{\frac{1}{3}} N^{\frac{2}{3}}$}
              \vspace{-10pt}
              \caption{The expected data persistency of Algorithm \ref{alg_symnew} for $p=1,$ $q=2,$ $r=1$}
              \label{fig:9}
          \end{figure}
      \end{minipage}
      \vspace{5pt}
  \end{minipage}    

Figure~\ref{fig:8} depicts the experimental data persistency of Algorithm~\ref{alg_rand} for $p=1,$ $q=2,$ $r=1,$ $D=N$
and the number of nodes $N\in\{48*1, 48*2, 48*3,\dots, 48*62\}$.
The additional line
$\left\{\left(N,\Gamma\left(\frac{4}{3}\right)N^{\frac{2}{3}}\right), 48\le N\le 2976\right\}$ is the leading term in the
theoretical estimation (see Theorem~\ref{thm:mainbd} for $p=1,$ $q=2,$ $r=1,$ $D=N$).
The experimental data persistency of Algorithm~\ref{alg_symnew} for $p=1,$ $q=2,$ $r=1,$ 
and the number of nodes $N\in\{48*1, 48*2, 48*3,\dots, 48*62\}$  is depicted in Figure~\ref{fig:9}.
Similary, the additional line
 $\left\{\left(N,\Gamma\left(\frac{4}{3}\right)3^{\frac{1}{3}} N^{\frac{2}{3}}\right), 48\le N\le 2976\right\}$ is the leading term in the
theoretical estimation (see Theorem~\ref{thm:mainbd} for $p=1,$ $q=2,$ $r=1$).
Figures~\ref{fig:8} and~\ref{fig:9} together illustrate the case when for both replication strategies random and symmetric the expected data persistency
is in $\Theta\left(N^{\frac{2}{3}}\right)$ (see Table~\ref{tab:dwasemen} in Section~\ref{sec:discussions}).

Finally, we point out that the carried experiments match the developed theory very well.
\section{Conclusions and Future research}
\label{sec:conclusions}
In this paper, we addressed the fundamental problem of the data persistency of replicated erasure codes for random and symmetric strategies.
To this end, we derived exact and asymptotically approximated formulas for the expected data persistency, that is, the number of nodes in the storage system that must be removed so that restoring some document would become impossible.

While we have discussed the applicability of our approach to non-uniform redundancy schemes, a natural open problem for future study is to find a closed form formula for the expected data persistency for non-uniform setting.
Additionally, it would be interesting to explore the problem of the data persistency
for a mixture of random and symmetric strategies or others strategies.
\bibliographystyle{plain}
\bibliography{P2P}
\end{document}